%% file: main_short.tex
\def\debug{1}

\documentclass[aps,amsmath,amssymb,reprint,superscriptaddress]{revtex4-2}
\usepackage{blindtext}
\usepackage{lipsum}
\usepackage{amsmath}
\usepackage{amsfonts}
\usepackage{amssymb}
\usepackage{comment}
\usepackage{graphicx}
\usepackage{dcolumn}
\usepackage{bm}
\usepackage{physics}
\usepackage[normalem]{ulem}
\usepackage{xcolor}
\usepackage{xr-hyper} 
\usepackage{xr} 
\usepackage[pdfborder={0 0 0}]{hyperref}
\usepackage{csquotes}
\usepackage{enumitem}
\usepackage{algorithm}
\usepackage{algpseudocode}

\usepackage{framed} 

\newcounter{algo}[section] 
\renewcommand{\thealgo}{\arabic{algo}}

\input{commands}

\usepackage[english]{babel}

\usepackage{amsthm}
\usepackage{bbm}
\usepackage{tocloft}
\usepackage{float}

\newtheorem{definition}{Definition}
\newtheorem{proposition}{Proposition}
\newtheorem{theorem}{Theorem}

\newtheorem{lemma}{Lemma}

\renewcommand{\thesection}{\Roman{section}} 
\begin{document}
\title{Linear Readout of Neural Manifolds with Continuous Variables}
\author{Will Slatton}
\thanks{These authors contributed equally to this work.}
\affiliation{Department of Physics and Kempner Institute, Harvard University, Cambridge, MA 02138, USA}
\author{Chi-Ning Chou}
\thanks{These authors contributed equally to this work.}
\affiliation{Center for Computational Neuroscience, Flatiron Institute, New York, NY 10010, USA}
\author{SueYeon Chung}
\affiliation{Department of Physics and Kempner Institute, Harvard University, Cambridge, MA 02138, USA}
\affiliation{Center for Computational Neuroscience, Flatiron Institute, New York, NY 10010, USA}

\date{\today}

\begin{abstract}
Brains and artificial neural networks compute with continuous variables such as object position or stimulus orientation.
However, the complex variability in neural responses makes it difficult to link internal representational structure to task performance.
We develop a statistical-mechanical theory of regression capacity that relates linear decoding efficiency of continuous variables to geometric properties of neural manifolds.
Our theory handles complex neural variability and applies to real data, revealing increasing capacity for decoding object position and size along the monkey visual stream.

\end{abstract}

\maketitle

\section{Introduction}
Navigating through space requires estimating continuous variables like spatial position from highly variable sensory inputs. For a neural population to usefully represent position, it must enable a downstream readout to decode this information while filtering out irrelevant stimulus features. A similar requirement extends to neural representations of any continuous variables.

Classic neuroscience experiments have revealed neurons tuned to continuous variables, including stimulus orientation in visual cortex ~\cite{hubel1965receptive}, spatial location in the hippocampus ~\cite{hafting2005microstructure,o1971hippocampus}, head direction~\cite{taube1990head}, and movement direction~\cite{georgopoulos1982relations}. Moving beyond single neuron tuning statistics \cite{seung1993simple}, recent advances in multi-neuron recordings have uncovered population-level structures shaped by mixed selectivity ~\cite{rigotti2013importance} and correlated variability ~\cite{averbeck2006neural}. These observations call for a theoretical framework that links the geometry of high-dimensional population codes to decoding efficiency of continuous variables.

Recently, neural manifolds—collections of neural population responses grouped by task-relevant variables—have emerged as a powerful abstraction for capturing such complex, structured variability~\cite{chung2021neural}.  Manifold capacity theory~\cite{chung2018classification,wakhloo2023linear,GLUE}, a recent theoretical framework using methods from statistical physics, quantifies how well a downstream neuron can distinguish discrete categories (e.g., pictures of cats versus pictures of dogs) based on the geometry of manifold-like neural activity patterns—also known as neural manifolds. This framework has been applied to both biological ~\citep{froudarakis2020object,yao2023transformation,paraouty2023sensory,hu2024representational} and artificial neural networks ~\citep{cohen2020separability,mamou2020emergence,stephenson2021geometry,dapello2021neural,kuoch2024probing,kirsanov2025geometrypromptingunveilingdistinct,chou2025featurelearninglazyrichdichotomy,chou2026diagnosing}, revealing how the shape and arrangement of neural manifolds determine the ability of downstream neurons to perform classification.

However, many tasks involve continuous variables (regression) rather than discrete classes (classification). Unlike for classification, no general theory connects manifold geometry to regression performance. Existing machine learning approaches often use strong parametric assumptions, offering little geometric insight into how population variability affects linear decodability~\cite{thrampoulidis2015regularized,barbier2019optimal,canatar2021spectral,canatar2024statistical,Bolle1993}.

\begin{figure}
    \centering
    \includegraphics[width=0.9\linewidth]{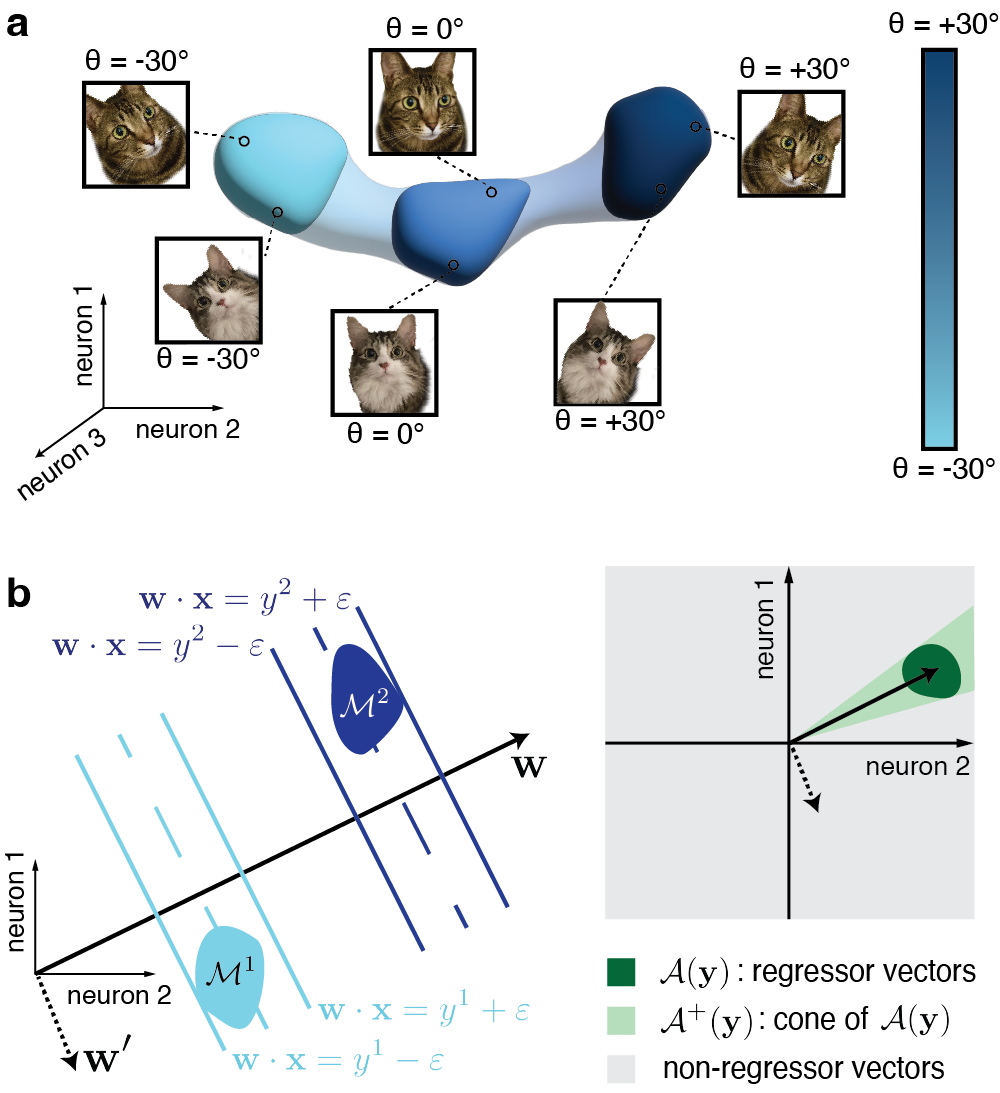}
    \caption{
    \textbf{a}, Neural manifolds arising from neural responses to images with orientation $\theta$ in $P=3$ disjoint stimulus intervals, recorded from $N=3$ neurons.
    \textbf{b}, The efficiency of manifold organization is assessed by testing if labels can be linearly read out. \textit{Left}: manifolds are $\varepsilon$-regressible if they lie within $-\varepsilon \leq \bw \cdot \bx - y^\mu \leq \varepsilon$, where $\bw$ is the regression vector. $\bw'$ is an example of non-regressor vector and $\cdot$ stands for inner product. \textit{Right}: Visualizing the solution space of regressor vectors.
    }
    \label{fig1}
\end{figure}

In this letter, we extend manifold capacity theory to regression. We develop a statistical mechanical theory for \emph{regression capacity}—a measure of linear readout performance for continuous variables (see \autoref{fig1})—and derive explicit formulas applicable to both synthetic models and data. For synthetic models such as spherical manifolds, we find closed-form solutions that link manifold geometry to readout performance. For example, regression capacity increases as the dimensionality and radius of the manifolds decrease.
We then apply our framework to a visual neuroscience dataset~\cite{Majaj2015}, demonstrating that representations of continuous object features become more efficient along the macaque ventral stream (\autoref{sec:hvm}). Our framework lays the foundation for identifying the geometric principles that govern the linear decodability of continuous variables in complex, high-dimensional datasets from neuroscience and machine learning.

\section{Theoretical framework}
Consider decoding a continuous variable (e.g., stimulus orientation;~\autoref{fig1}a) based on the activity of $N$ neurons. We consider $P$ discrete levels of the continuous variable, with labels $\{y^\mu \in \Real\}_{\mu=1}^P$. The collection of neural population responses mapped to a given bin, each a vector in the $N$-dimensional activity space, forms a manifold $\cM^\mu \subseteq \Real^N$.

By analogy to the classification theory~\cite{chung2018classification,GLUE}, we define the system's \textit{load} as the ratio $P/N$. As the load increases, the system becomes \enquote{less linearly regressible}—that is, the volume of admissible linear readout weights $\bm{w} \in \Real^N$ decreases. This motivates defining the regression capacity $\alpha$ of the system as \textbf{the maximal load under which a linear regressor still exists}. 
In the following, we study Gaussian mean-field models~\cite{chung2018classification,wakhloo2023linear,GLUE} for analytical characterization and instance-based models~\cite{GLUE} for data analysis.

\subsection{Mean-field theory for analytical study}\label{sec:mean-field capacity}
In the \textit{Gaussian mean-field models}~\cite{chung2018classification,wakhloo2023linear,GLUE}, we are interested in the thermodynamic limit ($P,N\to\infty$) to enable analytical characterization. We parameterize manifolds as $\cM^\mu = \{\bx^\mu=\bu^\mu_0+\sum_{i=1}^{D} s^\mu_i\bu^\mu_i\,:\,\bs^\mu\in\cS\}$ where $\bu^\mu_0\in\Real^N$ is the center of the $\mu$-th manifold, $\bu^\mu_1,\dots,\bu^\mu_{D}$ is a basis for the intrinsic part of the manifold, and $\cS\subset\Real^{D}$ is a $D$-dimensional convex set defining the manifold shape. For example, when the manifold is a $D$-dimensional L2 sphere with radius $R$, then $\cS = \{\bs\in\Real^D\, :\, \|\bs\|_2=R \}$.

The manifolds in the mean-field model are randomly embedded in an ambient space $\mathbb{R}^N$, with the embedding directions playing the role of quenched disorder in spin glass theory. The manifold bases $\bU=\{\bu_i^\mu\}_{\mu,i}$ are sampled from a correlated Gaussian distribution
\begin{equation}\label{eq:mean field model}
p(\bU) \propto \exp\left(-\frac{N}{2}\sum_{\mu,\nu,i,j,\ell}(\Sigma^{-1})^{\mu,i}_{\nu,j}u^\mu_{i,\ell}u^\nu_{j,\ell}\right)
\end{equation}
where $\Sigma\in\Real^{P(D+1)\times P(D+1)}$ is a covariance tensor, $\mu,\nu$ index the $P$ manifolds, $i,j$ index the $D$ intrinsic coordinates within the $\mu$-th and $\nu$-th manifold, and $\ell$ indexes the $N$ neurons. 
Finally, the manifold labels $\by=\{y^\mu\}$ are sampled from a distribution. While our capacity formula applies to arbitrary label distributions, in~\autoref{sec:synthetic} we focus on the case where the labels are drawn from a Gaussian distribution, which further simplifies the formula for geometric analysis.

To analytically study the regression capacity in the mean-field model, we consider the \textit{Gardner volume}~\cite{gardner1988space}, which quantifies the volume of linear readout vectors $\bm{w}$ that achieve regression error within a tolerance $\varepsilon\geq0$:
\begin{equation}
Z = \int_{\mathbb{R}^N} d\bw \prod_{\mu,\bx^\mu\in\cM^\mu}\Theta(\varepsilon-|\bw\cdot\bx^\mu-y^\mu|).
\end{equation}
The Heaviside function $\Theta(\cdot)$ enforces $|\bw \cdot \bx^\mu - y^\mu| \leq \varepsilon$, where $\bw\cdot\bx^\mu$ denotes the inner product between the two vectors and is the decoded label for $\bx^\mu$ and $y^\mu$ the target label. This approach follows earlier statistical mechanical theories of classification manifold capacity~\cite{chung2018classification,wakhloo2023linear,mignacco2025nonlinear,GLUE} and has been used in related classification settings~\cite{Bolle1993,Schnsberg2021,dahmen2020capacity,Clopath2012}. 

\begin{figure}
    \centering
    \includegraphics[width=\linewidth]{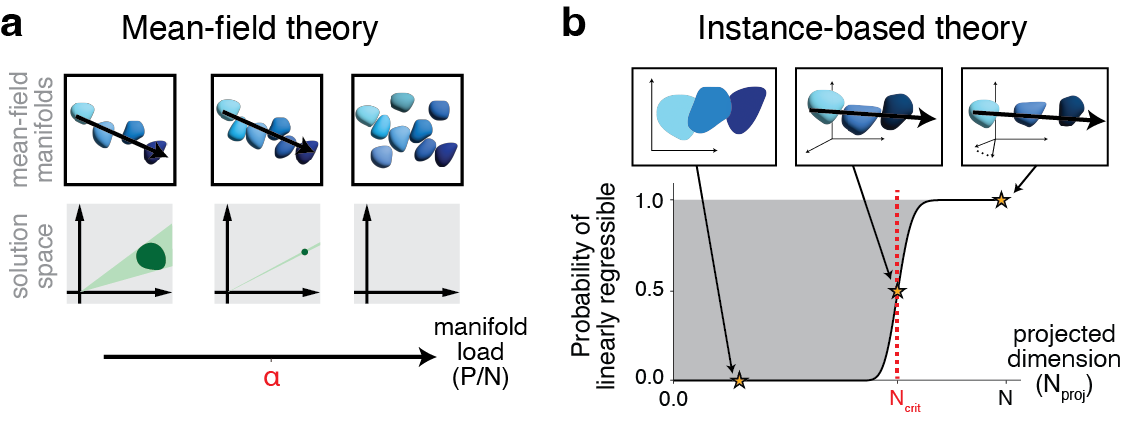}
    \caption{
    \textbf{a}, In mean-field theory, we randomly generate $P$ manifolds in $\Real^N$ (with $P,N\to\infty$), and define capacity as the maximum load $P/N$ while the Gardner volume (dark green area) is not zero.
    \textbf{b}, In instance-based theory, $P, N$ are finite and fixed (e.g., given by the data), and we consider randomly projecting the manifolds from $\Real^N$ to $\Real^{N_\proj}$, define critical dimension $N_\crit$ as the smallest $N_\proj$ with at least 0.5 probability of being linearly regressible (yellow star), and define capacity as $\alpha=P/N_\crit$.
    }
    \label{fig2}
\end{figure}

In the proportional limit where the load $P/N$ is fixed and $P,N\to\infty$, the typical Gardner volume exhibits a phase transition as the load increases, becoming zero above a certain critical load. The capacity $\alpha$ of the mean-field model is defined as this largest load with non-zero Gardner volume. By analyzing the disorder-averaged (over $\bU$ and $\by$) partition function $\overline{\log Z}$, we derive an analytical formula for the capacity:
\begin{equation}\label{eq:mean-field formula main}
\alpha_\mf(\varepsilon) = \lim_{P\to\infty}P\left(\Exp_{\bm{t},\bm{y}}[F_\mf(\bt,\by)]\right)^{-1},
\end{equation}
where $F_\mf(\bt,\by)=\min_{\bm{w}^+ \in \cA_\mf(\bm y)^+} \|\bm{t} - \bm{w}^+\|^2$ is a field related to the projection of $\bt$ to a convex cone $\cA_\mf(\bm y)^+ = \{\tau \bm w : \bm w \in \cA_\mf(\bm y), \tau \geq 0\}$ with $\cA_\mf(\bm y) = \bigcap_{\mu=1}^P \bigcap_{\bm{s} \in \mathcal{S}} \left\{\bm{w} \in \mathbb{R}^{P(D+1)} : \left\vert\bw\cdot((\Sigma^{1/2})^\mu\bs) - y^\mu\right\vert \leq \varepsilon\right\}$ being the set of all admissible linear regressor vectors. Here $(\Sigma^{1/2})^\mu$ denotes the $\mu$-th column block of the square root of the covariance matrix $\Sigma^{1/2}$. See the Supplementary Materials (SM) Section SII~\cite{SM} for details. The mean-field formula Eq.~\eqref{eq:mean-field formula main} enables analytical calculation of regression capacity for the synthetic models studied in~\autoref{sec:examples}.

\subsection{Instance-based theory for data analysis}\label{sec:instance-based capacity}
In the \textit{instance-based model}~\cite{GLUE} for data analysis, the number of experimental conditions $P$ and  recorded units $N$ are finite. Let $\bx_1^\mu,\dots,\bx_M^\mu\in\Real^N$ be the neural activity vectors for the $\mu$-th experimental condition (e.g., $M$ trials with the same stimulus angle $\theta^\mu$ but with other stimulus features potentially varying). The neural manifold for this condition is the set $\cM^\mu:=\{\bx^\mu_i\}_{i=1}^M$ and it has some label $y^\mu \in \Real$ (e.g., $y^\mu=\theta^\mu$ for decoding stimulus angle). See~\autoref{sec:hvm} for an example on a monkey vision dataset.

We define instance-based regression capacity as $P/N_\crit$, where $N_\crit$ is the \textit{critical dimension}—the smallest dimension $N_\proj$ such that, after random projection to an $N_\proj$-dimensional subspace, at least one admissible linear readout exists with probability $\geq 1/2$ (see \autoref{fig2}b). We derive a closed-form estimator for this instance-based capacity:
\begin{equation}\label{eq:alphaib}
\alpha_\ib(\varepsilon) = P\left(\Exp_{\bm{t}}[F_\ib(\bt)]\right)^{-1},
\end{equation}
where $F_\ib(\bt)=\min_{\bm{w}^+ \in \cA_\ib^+} \|\bm{t} - \bm{w}^+\|^2$ is a field related to the projection of $\bt$ to a convex cone $\cA_\ib^+$ generated by set of admissible readout weights $\cA_\ib = \bigcap_{\mu=1}^P \bigcap_{\bm{x}^\mu \in \mathcal{M}^\mu} \left\{\bm{w} \in \mathbb{R}^{N} : \left\vert \bm{w}^\mu\cdot\bx^\mu - y^\mu\right\vert \leq \varepsilon\right\}$.
That is, $\cA_\ib^+ = \{\tau \bm w : \bm w \in \cA_\ib, \tau \geq 0\}$. 
We defer the details and other extensions (e.g., introducing a bias term in the linear regression) to SM Section SI~\cite{SM}. An example of applying the instance-based capacity to real dataset is given in~\autoref{sec:hvm}. We remark that the capacity formula Eq.~\eqref{eq:alphaib} can be empirically estimated in quadratic time (in $N,P$) via a standard QP solver. 

\section{Examples}\label{sec:examples}
We now apply our regression capacity framework to obtain closed-form capacity expressions for a handful of synthetic manifold models and interpret the results in~\autoref{sec:synthetic}. Next, we use instanced-base regression capacity to study a neuroscience dataset in~\autoref{sec:hvm}.

\subsection{Applications to synthetic manifolds}\label{sec:synthetic}
We start by using the mean-field formula Eq.~\eqref{eq:mean-field formula main} to compute capacity for synthetic data models and study how their geometric properties (e.g., correlations, radius, and dimensionality) influence capacity.

\paragraph{Point-like manifolds}\label{sec:point}
First we consider $P$ manifolds, each consisting of a single point. The $P$ points are sampled i.i.d.~from $\bx^\mu \sim \mathcal{N}(0,r^2I_N/N)$ with i.i.d.~labels $y^\mu\sim \mathcal{N}(0, \sigma^2)$ (see~\autoref{fig:points}a). The parameters $r,\sigma$ quantify the overall scale of the points and labels respectively. Starting with the mean-field formula Eq.~\eqref{eq:mean-field formula main}, we calculate the regression capacity of these uncorrelated points:
\begin{equation}\label{eq:uncorr-points-cap}
    \alpha^{-1}(\varepsilon) = 2 \min_{\tau \geq 0} \int_{b/a}^\infty \mathcal{D} z (az-b)^2,
\end{equation}
where $\mathcal{D} z$ is the standard Gaussian measure, $a=\sqrt{1 + \sigma^2\tau^2/r^2}$, $b=\varepsilon\tau/r$ (derived in SM~\cite{SM} Sec.~III), and $\tau$ is a scale factor naturally appears from the derivation.

We note that: (i) $\alpha(0) = 1$, consistent with the fact that $P \leq N$ linearly independent points can be mapped to arbitrary labels; (ii) $\alpha(\varepsilon)$ increases monotonically with tolerance; (iii) varying $r$ or simultaneously scaling $\varepsilon$ and $\sigma$ leaves capacity unchanged. Thus, the regression capacity of uncorrelated points depends only on the rescaled parameter $\varepsilon_{\textsf{equiv}} := \varepsilon/\sigma$, which we call the \textit{equivalent tolerance} (the tolerance of an equivalent model with $\sigma = 1$).

\begin{figure}
    \centering
    \includegraphics[width=1.0\linewidth]{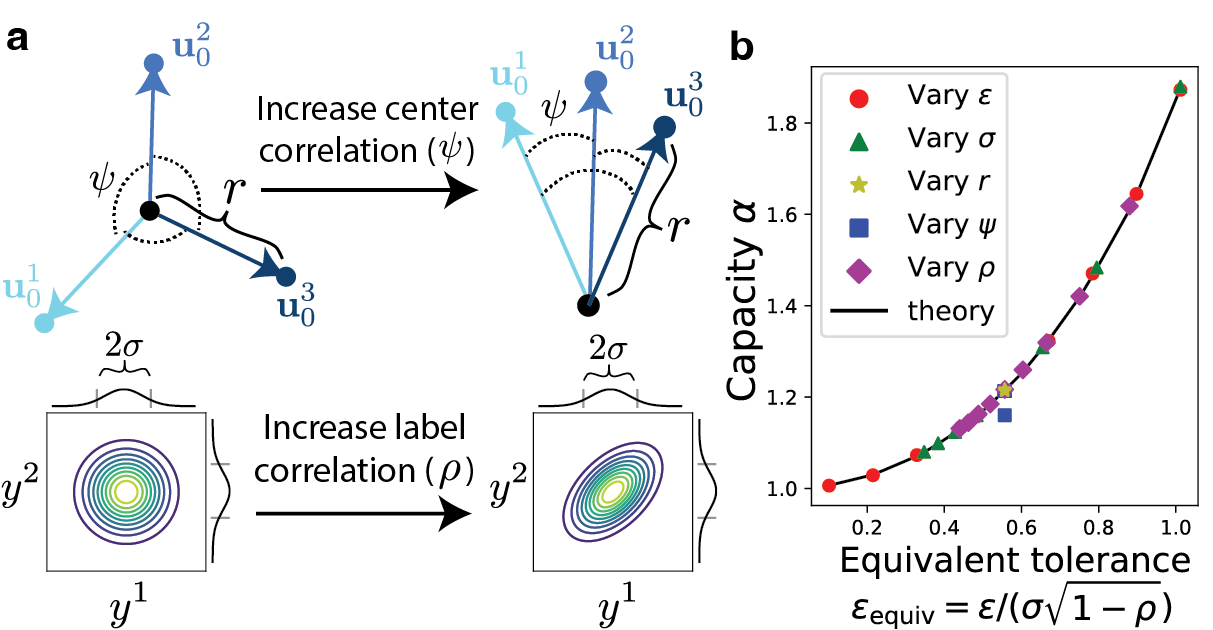}
    \caption{Mean-field point-like manifolds.  
    \textbf{a}, $r$ controls the average norm of the points $\{\mathbf{u}_0^\mu\}$, $\sigma$ controls the scale of the target labels $\{y^\mu\}$, and $\psi,\rho \in [0,1)$ set the correlation strengths of the data points and labels respectively.
    \textbf{b}, Varying one of $\varepsilon, \sigma, r, \psi, \rho$ (others fixed), we numerically estimating capacity with $P=500$ manifolds (see SM~\cite{SM}). The black line shows the theoretical capacity from Equation Eq.~\eqref{eq:uncorr-points-cap} in the $P \to \infty$ limit, where capacity depends monotonically only on the \emph{asymptotically equivalent tolerance} $\varepsilon_\mathrm{equiv} = \varepsilon/(\sigma\sqrt{1-\rho})$.
    }
    \label{fig:points}
\end{figure}

We next introduce correlations among points and labels. The data and label covariances are set to $\Sigma_{\mu,0}^{\nu,0} = r^2[(1-\psi)\delta_{\mu\nu}+\psi]$ and $\Lambda_{\mu}^{\nu} = \sigma^2[(1-\rho)\delta_{\mu\nu}+\rho]$, where $\psi \in [0, 1)$ and $\rho \in [0, 1)$ control correlation strength among data and labels respectively (see~\autoref{fig:points}a). As shown in SM~\cite{SM} Sec.~IV, the resulting capacity is asymptotically equivalent to the uncorrelated case with rescaled center norm $r' = r\sqrt{1-\psi}$ and label scale $\sigma' = \sigma\sqrt{1-\rho}$. Thus, regression capacity again depends only on the equivalent tolerance $\varepsilon_{\textsf{equiv}} = \varepsilon/(\sigma\sqrt{1-\rho})$ in the $P\to\infty$ limit, as verified in~\autoref{fig:points}b.

\begin{figure}
    \centering
    \includegraphics[width=1.0\linewidth]{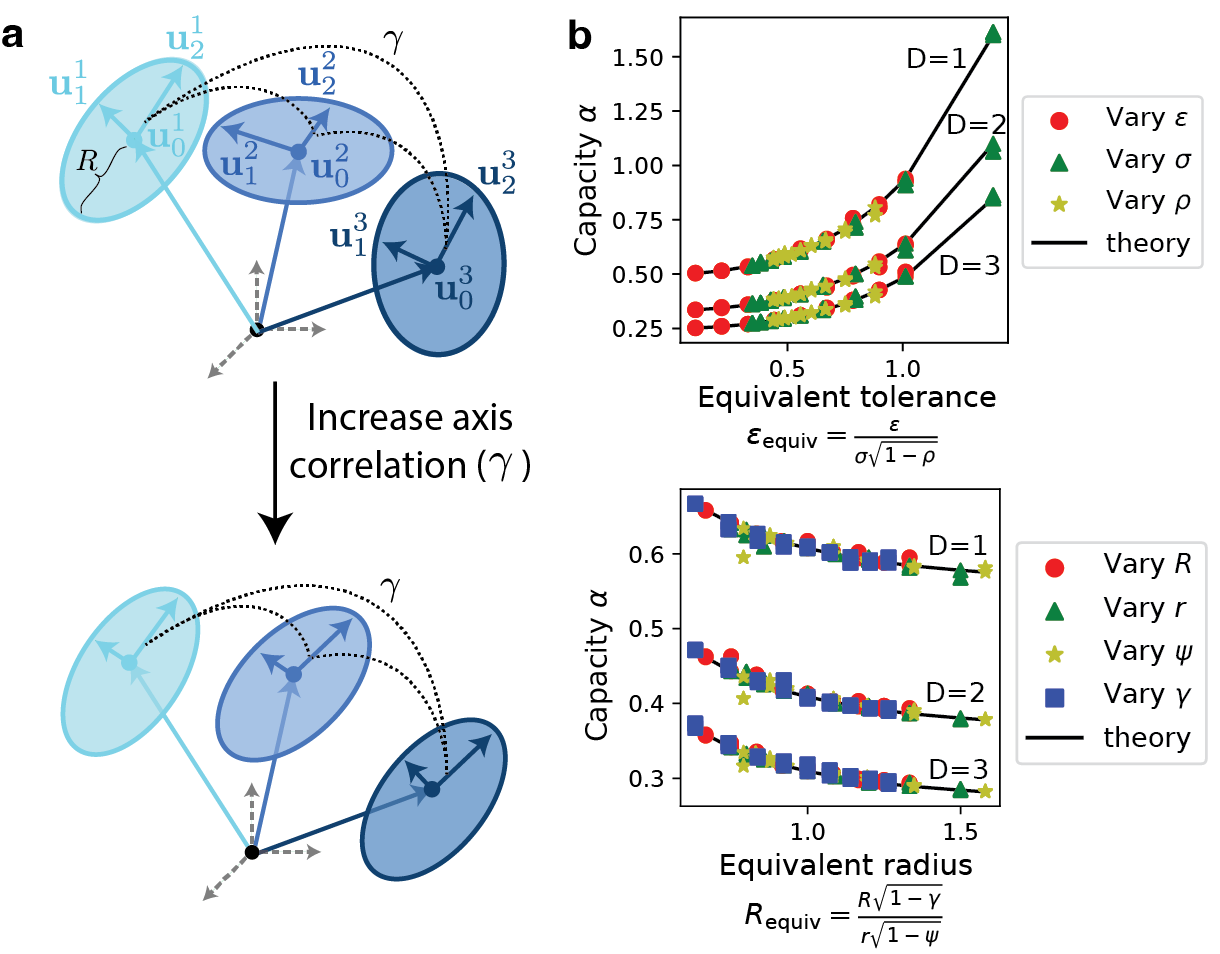}
    \caption{Mean-field sphere-like manifolds.  
    \textbf{a}, Manifold centers and labels are drawn from the correlated points model in~\autoref{fig:points}. $R$ sets the mean radius; $\gamma\in [0,1)$ controls axis correlations. 
    \textbf{b}, Theory–numeric check. Solid lines show the analytic formula Eq.~\eqref{eq:corrspherecap}, and dots indicate numerical estimates with $P=250$ manifolds (see Section VII). Each curve corresponds to a different $D$, with capacity decreasing in $D$.
    \textit{Top}: capacity depends on $\varepsilon, \sigma, \rho$ only through the \emph{asymptotically equivalent tolerance} $\varepsilon_\mathrm{equiv} = \frac{\varepsilon}{\sigma\sqrt{1-\rho}}$ and increases with $\varepsilon_\mathrm{equiv}$. 
    \textit{Bottom}: capacity depends on $R, r, \psi, \gamma$ only through the \emph{asymptotically equivalent radius} $R_\mathrm{equiv} = \frac{R\sqrt{1-\gamma}}{r\sqrt{1-\psi}}$ and decreases with $R_\mathrm{equiv}$.
    }
    \label{fig:corr-spheres}
\end{figure}

\paragraph{Spherical manifolds}\label{sec:sphere}
We next introduce structured variability in the form of low-rank spherical manifolds. We consider $D$-dimensional spherical shapes $\mathcal{S} = \{\bs \in \mathbb{R}^D : \|\bs\|_2 = 1\}$, so that the $\mu$-th manifold is $\mathcal{M}^\mu = \left\{\bu_0^\mu + \sum_{i=1}^D s_i^\mu \bu_i^\mu : s^\mu \in \mathcal{S}\right\}$.
The correlations among manifold centers $\bu_0^\mu$ and axes $\bu_i^\mu$ ($i>0$) are modeled by a correlation tensor $\Sigma^{\nu,j}_{\mu,i}$. We sample $\bu_i^\mu$ as $N$-dimensional Gaussians with $\Exp[\langle \bu_i^\mu, \bu_j^\nu \rangle]=\Sigma_{\nu,j}^{\mu,i}$ where
$$\Sigma_{\mu,i}^{\nu,j} = \begin{cases}
    r^2[(1-\psi)\delta_{\mu\nu}+\psi] & \textrm{if } i = j = 0, \\
    R^2[(1-\gamma)\delta_{\mu\nu}+\gamma] & \textrm{if } i = j > 0, \\
    0 & \textrm{else,}
\end{cases}$$
and sample correlated Gaussian labels $\mathbf{y} \sim \mathcal{N}(0, \Lambda)$ with covariance matrix $\Lambda_\mu^{\nu} = \sigma^2[(1-\rho)\delta_{\mu\nu}+\rho]$. The parameters $r,R,\sigma$ scale the norm of the manifold centers, the size of the spherical variability around each manifold center, and the labels respectively. The parameters $\psi,\rho,\gamma\in[0,1)$ control the respective strengths of the center, label, and axis correlations (\autoref{fig:corr-spheres}a).

We analytically solve the regression capacity of this model of spherical manifolds:
\begin{equation}\label{eq:corrspherecap}
\alpha^{-1}(\varepsilon) = 2\min_{\tau \geq 0} \int_{0}^\infty dx \; \chi_D(x) g(x, \tau, \varepsilon),
\end{equation}
where $g$ can be written in closed-form (see SM~\cite{SM}) and $\chi_D(\cdot)$ is the probability density function of the chi distribution with $D$ degrees of freedom (see Sec.~IV of SM~\cite{SM} for a derivation). The inner function $g(x, \tau, \varepsilon)$ is complex, but it is invariant under several parameter changes. For example, $(\varepsilon,\sigma,\tau)\leftrightarrow(a\varepsilon, a\sigma, \tau/a)$: Invariance to multiplying tolerance, label scale, and readout scale (equivalent to \emph{dividing} $\tau$) by some constant. See SM~\cite{SM} for a complete list of symmetries in this model. From the symmetries of $g(x, \tau, \varepsilon)$, it follows that the capacity depends only on the \textit{asymptotically equivalent tolerance} $\varepsilon_\mathrm{equiv} := \varepsilon/(\sigma\sqrt{1-\rho})$, the \textit{asymptotically equivalent radius} $R_\mathrm{equiv} := R\sqrt{1-\gamma}/(r\sqrt{1-\psi})$, and the sphere dimension $D$. We numerically verify this in~\autoref{fig:corr-spheres}.

\begin{figure}
    \centering
    \includegraphics[width=1.0\linewidth]{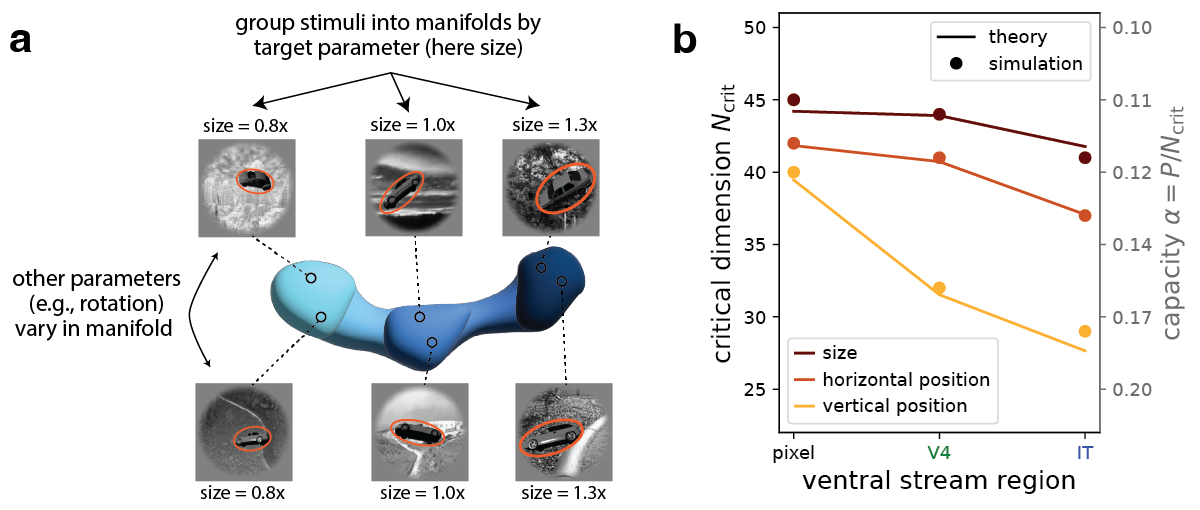}
    \caption{Application to neuroscience data.
    \textbf{a}, We analyze electrophysiological recordings from macaque ventral stream as monkeys viewed objects with variable pose parameters (e.g., size, position) superimposed on complex backgrounds.
    \textbf{b}, For each target parameter (e.g., size), we evenly bin its range into 5 intervals, yielding $P = 5$ manifolds. Each manifold consists of neural responses to 15 individual stimuli whose non-target pose parameters can vary freely, contributing to neural response variability. The background also varies within each manifold and provides another source of nuisance variability. We observe that capacity for each variable increases (i.e., $N_\mathrm{crit}$ decreases) from pixels to V4 to IT, indicating more efficient representations at higher processing stages. Our analytical estimates (solid lines) match numerical results (dots).
    }
    \label{fig:hvm}
\end{figure}

\subsection{Application to neuroscience data}\label{sec:hvm}
Finally, we demonstrate that our regression capacity theory can be applied to quantify the linear decodability of continuous latent variables from real data with complex variability. We apply our method to electrophysiological recordings from the macaque ventral stream~\cite{Majaj2015}, and study representations of object pose parameters (e.g., the relative size, horizontal position, and vertical position) of a foreground object on a variable background (\autoref{fig:hvm}a). Conceptually, a successful decoder for one latent variable must filter out the neural variability induced by intrinsic noise and the encoding of all the other stimulus features. Previous work quantified decodability of object pose parameters from the same dataset by the conventional method of training a linear decoder on the full-dimensional data and evaluating its cross-validated generalization error~\cite{Hong2016}. Our results in~\autoref{fig:hvm}b qualitatively match the standard decoding analysis in finding monotonic improvement in decodability from raw image pixels to area V4 to area IT for all three pose parameters. Furthermore, unlike generalization error-based metrics (e.g., support vector regression~\cite{drucker1996support,canatar2024statistical}) used in previous work~\cite{Hong2016}, the numerical value of our regression capacity measure can be directly interpreted as the number of neurons a downstream reader needs to decode the target variable to a desired accuracy $\varepsilon$.

One aspect of manifold construction worth noting is that we estimate capacity separately for each top-level object class in the dataset (e.g., ``animal'' and ``fruit''), since objects belonging to these very different categories are unlikely to share a meaningful notion of absolute rotation angle. The values in \autoref{fig:hvm}b are averaged across top-level categories. 

\section{Discussion}
In this Letter, we introduced a theoretical framework for quantifying linear decodability of continuous variables encoded in neural population manifolds, extending manifold capacity theory from classification tasks to regression problems. Previous studies of estimation and fine discrimination from neural population activity focused on statistical properties of noise~\cite{seung1993simple, dayan2005theoretical, averbeck2006neural}, often assuming Gaussian or Poisson distributions. However, neural activity often exhibits more complex structured variability from coding nuisance variables (e.g., background features), which significantly impacts readout performance~\cite{seung1993simple, mato1996neural, DICARLO2007333, chung2021neural}. Characterizing such complex structures has been challenging due to limited methods capable of capturing their geometry~\cite{chung2021neural}. 

By generalizing the previous manifold capacity theory from discrete to continuous labels, we systematically relate structured variability to decoding efficiency for continuous variables. We explicitly link geometric properties of neural  manifolds—such as dimensionality, radius, and correlation structures—to their regression capacities, showing that uniform correlations among manifolds and their labels effectively just rescale the data (\autoref{sec:synthetic}).

We also show how to non-parametrically define and estimate regression capacity, so that it can be directly quantified in real neuroscience datasets. We apply our framework to recordings from primate visual cortex to validate its utility beyond synthetic scenarios, finding increases in linear decodability of object pose parameters along the ventral stream consistent with past analysis~\cite{Hong2016}. These findings support the view that neural networks, whether biological or artificial, hierarchically refine representations to optimize downstream decoding. 

Our framework opens several directions for technical and applied extensions. 
On the applied side, our methods offer a principled approach for analyzing neural data across sensory modalities, extending prior work in the classification setting~\cite{GLUE} to a wide range of continuous-variable tasks, including navigation and estimation problems in perception and motor control. In neuroscience, this enables quantitative comparisons of representational efficiency across brain areas or behavioral conditions, even in the presence of structured noise and nuisance variability. In machine learning, it provides a geometric lens to characterize how deep networks organize task-relevant continuous variables—across layers or throughout learning.

\section*{Acknowledgments}
This work was partially supported by the Center for Computational Neuroscience at the Flatiron Institute of the Simons Foundation, and the Kempner Institute for the Study of Natural and Artificial Intelligence at Harvard University. S.C. is partially supported by the Klingenstein-Simons Award, a Sloan Research Fellowship, NIH award R01DA059220, and the Samsung Advanced Institute of Technology (under the project ``Next Generation Deep Learning: From Pattern Recognition to AI''). All experiments were performed on the Flatiron Institute high-performance computing cluster.

\nocite{Rockafellar1996-av}


\bibliographystyle{apsrev4-1} 
\bibliography{biblio} 

\onecolumngrid

\renewcommand{\theequation}{S\arabic{equation}}
\renewcommand{\thetable}{S\arabic{table}}
\renewcommand{\thefigure}{S\arabic{figure}}
\renewcommand{\thesection}{S\Roman{section}}
\setcounter{equation}{0}
\setcounter{table}{0}
\setcounter{figure}{0}
\setcounter{section}{0}

\newpage

\input{supmat}

\end{document}

%% file: commands.tex
\usepackage{mathtools}

\newcommand{\cM}{\mathcal{M}}
\newcommand{\cS}{\mathcal{S}}
\newcommand{\cA}{\mathcal{A}}

\newcommand{\cN}{\mathcal{N}}
\newcommand{\Real}{\mathbb{R}}
\newcommand{\bs}{\mathbf{s}}
\newcommand{\bt}{\mathbf{t}}
\newcommand{\bv}{\mathbf{v}}
\newcommand{\bw}{\mathbf{w}}
\newcommand{\bx}{\mathbf{x}}
\newcommand{\by}{\mathbf{y}}
\newcommand{\bz}{\mathbf{z}}
\newcommand{\bu}{\mathbf{u}}
\newcommand{\bU}{\mathbf{U}}

\newcommand{\Exp}{\mathop{\mathbb{E}}}
\newcommand{\proj}{\mathrm{proj}}
\newcommand{\simcap}{\mathrm{simulation}}
\newcommand{\mf}{\textsf{mf}}
\newcommand{\ib}{\textsf{ib}}
\newcommand{\crit}{\textsf{crit}}

\newcommand{\Cnotes}[1]{\ifnum\debug=1{\color{red} [CNC: #1]}\fi}
\newcommand{\Snotes}[1]{\ifnum\debug=1{\color{magenta} [SYC: #1]}\fi}
\newcommand{\Wnotes}[1]{\ifnum\debug=1{\color{olive} [WS: #1]}\fi}

%% file: supmat.tex
\begin{center}
\bf\large Supplementary Materials
\end{center}
\begin{center}
{\it Linear Readout of Neural Manifolds with Continuous Variables}
\end{center}

\setcounter{page}{1}




The Supplementary Materials (SM) are organized in five parts: in~\autoref{sec:instance-based-cap} and~\autoref{sec:si mean-field} we present the details for our instance-based regression capacity theory and our mean-field regression capacity theory respectively. In~\autoref{sec:si point-like manifolds} and~\autoref{sec:si sphere-like manifolds} we present the details for analyzing point-like manifolds and sphere-like manifolds respectively. 

\def\arraystretch{1.2}
\begin{table}[ht!]
\centering
\begin{tabular}{ll}
\hline
\underline{Background:} & \\
$\Real, \Real_{\geq0}$ & The set of real numbers and the set of non-negative real-numbers\\
$\Exp;\ \mathbb{P}$ & Expectation operator; Probability of an event\\
$\cN(\cdot,\cdot)$ & Multivariate normal distribution; first argument: mean vector;\\
&second argument: covariance matrix \\
a, $\mathbf{a}$ & non-boldface: scalar variable; boldface: vector/function variable \\ 
$\|\cdot\|_2$ & L2 norm of a vector \\
$\varepsilon$ & Tolerance parameter \\ 
$\varnothing$ & Empty set \\
$N,P$ & The number of neurons and the number of manifolds \\
$\bw\in\Real^N$ & Readout weight vector \\ 
$\cM^\mu\subset\Real^N$ & The $\mu$-th manifold \\ 
$\bx^\mu\in\cM^\mu$ & A point in the $\mu$-th manifold \\ 
$y^\mu\in\Real$ & The label for the $\mu$-th manifold \\ 
$\cA,\cA^+$ & A set of admissible regression weights; the convex cone of $\cA$ \\ [3pt]
\underline{Mean-field theory:} & \\
$0\leq N_\proj\leq N$ & The dimensionality of neural activity space after a projection \\
$N_\mf(\varepsilon)$ & Mean-field formula for critical dimension \\
$\alpha_\mf(\varepsilon)$ & Mean-field capacity formula \\ [3pt]
\underline{Instance-based theory:} & \\
$\Pi^{(N_\proj)}$ & Random projection operator from $\Real^N$ to $\Real^{N_\proj}$ \\
$p(N_\proj)$ & Probability of regressibility after random projection to $\Real^{N_\proj}$ \\
$N_\crit(\varepsilon)$ & Empirical definition for critical dimension \\
$\alpha_\simcap(\varepsilon)$ & Empirical definition for regression capacity \\
$N_\ib(\varepsilon)$ & Instance-based formula for critical dimension \\
$\alpha_{\ib}(\varepsilon)$ & Instance-based capacity formula \\ [3pt]
\underline{Synthetic models:} & \\
$\psi,\rho,\gamma$ & Correlations between manifold center, label, and axes \\ [3pt]
\hline
\end{tabular}
\caption{Notations.}
\label{table:notations}
\end{table}

\section{Instance-Based Regression Capacity Theory}\label{sec:instance-based-cap}
In this section, we provide details for the instance-based regression capacity as introduced in~\autoref{sec:instance-based capacity} in the main text. For a given high-dimensional neural data with $N$ neurons and $P$ manifolds (where both $N,P$ being finite), we directly define its manifold capacity based on some random projection procedures. This version of capacity has the distinct advantage of \emph{not requiring an explicit generative model} for the data to be regressed and can therefore be directly evaluated on real data for a which a generative model is unknown. In Section II, we will use this \enquote{instance-based} theory to recover a more standard mean-field notion of regression capacity for a generative model with manifolds of an given shape embedded with some specified second-order correlation structure.

\subsection{Definitions of simulation capacity and useful mathematical concepts}
Suppose we have $P$ fixed manifolds $\{\mathcal{M}^\mu \subseteq \mathbb{R}^N\}_{\mu=1}^P$ that are embedded in a shared $N$-dimensional state space and that have some target labels $\by=\{y^\mu \in \mathbb{R}\}_{\mu=1}^P$ for regression. As we focus on linear regression on top of these manifolds, without loss of generality we can treat these manifolds as convex sets. These manifolds may be correlated in any arbitrary way with each other and with the labels and may all be uniquely shaped. For an error tolerance $\varepsilon \geq 0$ for which an admissible readout $\mathbf{w} \in \mathbb{R}^N$ exists. By \emph{admissible}, we mean that
$$|\mathbf{w} \cdot \mathbf{x}^\mu - y^\mu| \leq \varepsilon \textrm{ for all } \mu \in \{1, \dots, P\} \textrm{ and } \mathbf{x}^\mu \in \mathcal{M}^\mu.$$
For any dimension $N_\mathrm{proj}$ with $0 \leq N_\mathrm{proj} \leq N$, we can randomly project the data manifolds down to $\mathbb{R}^{N_\mathrm{proj}}$ and check whether it remains regressible (i.e., an admissible weight exists). Let $p(N_\mathrm{proj})$ be the probability of the data still being regressible when projected to $\mathbb{R}^{N_\mathrm{proj}}$. Since $p(0) = 0$ and $p(N) = 1$, even for finite-dimensional data there will be some value $N_\crit$ such that $p(N_\crit) \leq \frac{1}{2}$ but $p(N_\crit + 1) > 1/2$. 
We can then define the \enquote{instance-based} regression capacity $\alpha_\mathrm{simulation} := P/N_\crit$ for our dataset. This is illustrated in SM Figure \ref{fig:ddcap} and formalized in the following definition:

\begin{figure}
    \centering
    \includegraphics[width=0.8\linewidth]{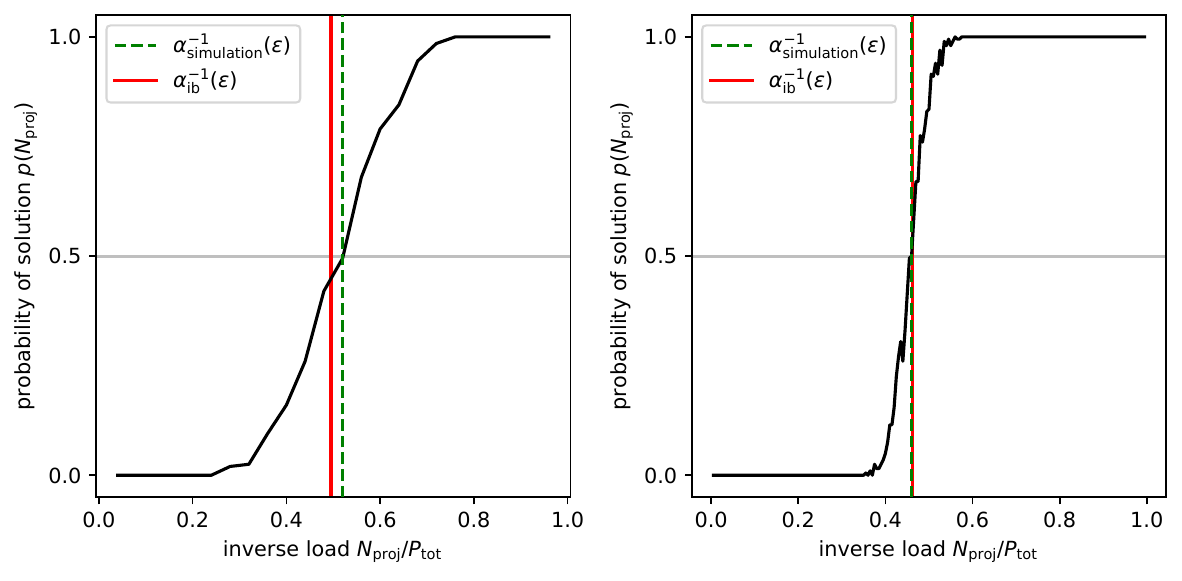}
    \caption{Numerical illustration of instance-based capacity for finite-dimensional data. Both (a) and (b) use manifolds formed from the raw pixel-level stimuli from the macaque vision dataset used in Figure \ref{fig:hvm} and use vertical position as the target labels as an arbitrary example of realistic data. (a) In the first panel, we study $N = 25$ dimensional representations (randomly projected from the full $256^2$-dimensional grayscale image pixels) with $P = 5$ manifolds consisting of $5$ points each, yielding $P_\mathrm{tot} = 25$ total points. We can numerically compute the instance-based capacity $\alpha_\mathrm{simulation}(\varepsilon)$ for an example tolerance $\varepsilon = 0.3$ by bisecting the range $0 \leq N_\mathrm{proj} \leq N$ and finding the unique $N_\mathrm{proj}$ such that $p(N_\mathrm{proj}) \leq \frac{1}{2}$ and $p(N_\mathrm{proj} + 1) > \frac{1}{2}$, even though there is not a sharp phase transition from $p(N_\mathrm{proj}) \approx 0$ to $p(N_\mathrm{proj}) \approx 1$ for this relatively low-dimensional data. Even in this case, the capacity estimator $\alpha_\mathrm{ib}(\varepsilon)$ from Theorem \ref{ddcap} is a good approximation of $\alpha_\mathrm{simulation}(\varepsilon)$. (b) In the second panel, we study higher-dimensional data with $N = 200$ and 40 points per manifold, yielding $P_\mathrm{tot} = 200$ total points. We start to approach a sharp phase transition and $\alpha(\varepsilon)$ almost perfectly approximates $\alpha_\mathrm{simulation}(\varepsilon)$, which gives the location of the phase transition.}
    \label{fig:ddcap}
\end{figure}

\begin{definition}[Simulation capacity]\label{def:ddcap}
    Let $P$ manifolds $\{\mathcal{M}^\mu \subseteq \mathbb{R}^N\}_{\mu=1}^P$ and labels $\{y^\mu \in \mathbb{R}\}_{\mu=1}^P$ be given as above. For any tolerance $\varepsilon \geq 0$ and any matrix $M \in \mathbb{R}^{N_\mathrm{proj} \times N}$, we define the following concepts:
    \begin{itemize}
    \item (Admissible regression weights) Define the set $\mathcal{A}(M)$ of admissible regression weights for the data transformed by $M$:
    $$\mathcal{A}(M) = \left\{\bm{w} \in \mathbb{R}^{N_\mathrm{proj}} : |\bm{w} \cdot (M\bm{x}^\mu) - y^\mu| \leq \varepsilon \textrm{ for all } \mu \in \{1, \dots, P\}, \bm{x}^\mu \in \mathcal{M}^\mu\right\}.$$
    \item (Regressible probability) Assume that the original data is regressible within the specified tolerance (i.e., that $\mathcal{A}(I_N) \neq \varnothing$) and that no constant readout is admissible (i.e., that $0 \not\in \mathcal{A}(I_N))$. Then given any integer $N_\mathrm{proj}$ with $0 \leq N_\mathrm{proj} \leq N$, define
    \[
    p(N_\mathrm{proj}) = \mathbb{P}_{\Pi^{(N_\mathrm{proj})}}\left( \mathcal{A}(\Pi^{(N_\mathrm{proj})}) \neq \varnothing \right) \, ,
    \]
    where $\Pi^{(N_\mathrm{proj})} \in \mathbb{R}^{N_\mathrm{proj} \times N}$ is a uniformly random projection matrix.\footnote{By \enquote{uniformly random projection matrix}, we mean that $\Pi^{(N_\mathrm{proj})}$ is the first $N_\mathrm{proj}$ columns of an orthogonal matrix sampled from the normalized Haar measure on $O(N)$.}
    Note that $p(N_\mathrm{proj})$ is a non-decreasing function with $p(0) = 0$ and $p(N) = 1$. 
    \item (Critical dimension) Define the instance-based critical dimension of the dataset $\{\mathcal{M}^\mu\},\{y^\mu\}$ to be
    \[
    N_\crit(\varepsilon) := \min_{p(N_\proj)\geq0.5}\{N_\proj\}
    \]
    \item (Simulation capacity) Finally, we define the instance-based regression capacity of the dataset $\{\mathcal{M}^\mu\},\{y^\mu\}$ to be 
    \[
    \alpha_\mathrm{simulation}(\varepsilon) = \frac{P}{N_\crit(\varepsilon)} \, .
    \]
    \end{itemize}
\end{definition}

The above notion of instance-based capacity is weaker than the usual asymptotic definitions, since it does not require a sharp phase transition from the data almost never being regressible ($p(N_\mathrm{proj}) \approx 0$) to almost always being regressible ($p(N_\mathrm{proj}) \approx 1$) at the critical $N_\crit(\varepsilon) = P/\alpha_\simcap(\varepsilon)$. However, the following theorem shows that as the data dimensionality $N$ goes to infinity, $p(N_\mathrm{proj})$ approaches such a step function. The theorem also gives a more direct estimator for capacity that agrees with the instance-based capacity $\alpha_\simcap(\varepsilon)$ in the thermodynamic limit $N \to \infty$ (and numerically agrees quite well with it even for finite data, as shown in SM Figure \ref{fig:ddcap}).

\subsection{Formulas for critical dimension and instance-based capacity}
Now, we are ready to present our mathematical formulas for critical dimension and simulation capacity (\autoref{def:capformula}).

\begin{definition}[Capacity formula]\label{def:capformula}
    Let $P$ manifolds $\{\mathcal{M}^\mu \subseteq \mathbb{R}^N\}_{\mu=1}^P$ and labels $\{y^\mu \in \mathbb{R}\}_{\mu=1}^P$ be given as above. For any tolerance $\varepsilon \geq 0$, define an estimator $N_\ib(\varepsilon)$ for the critical dimension $N_\crit(\varepsilon)$ as
    \[
    N_\ib(\varepsilon) := \Exp_\bt\left[\min_{\bw^+\in\cA^+}\|\bt-\bw^+\|_2^2\right] \, ,
    \]
    where $\mathcal{A} = \mathcal{A}(I_N)$ is the set of admissible regression weights for the original data and $S^+ := \{\lambda \mathbf{s} : \mathbf{s} \in S, \lambda \geq 0\}$ denotes the cone generated by a set $S \subseteq \mathbb{R}^N$. Define an estimator $\alpha(\varepsilon)$ for simulation capacity $\alpha_\simcap(\varepsilon)$ as
    \[
    \alpha_\ib(\varepsilon) := \frac{P}{N_\ib(\varepsilon)} \, .
    \]
\end{definition}

In the following, we use mathematical tools from large deviation theory for convex geometry to show that $N_\ib(\varepsilon)$ and $\alpha_\ib(\varepsilon)$ converge to $N_\crit(\varepsilon)$ and $\alpha_\simcap(\varepsilon)$ respectively in the thermodynamic limit (i.e., $N\to\infty)$.

\begin{theorem}[Formulas for critical dimension and instance-based capacity]\label{ddcap}
    Given $P$ convex sets $\{\mathcal{M}^\mu \subseteq \mathbb{R}^N\}_{\mu=1}^P$ as data manifolds and (deterministic) target labels $\{y^\mu\}_{\mu=1}^P$. Let $\cA(\cdot),p(\cdot),N_\crit(\cdot),\alpha_\simcap(\cdot)$ as defined in~\autoref{def:ddcap} and $N_\ib(\cdot),\alpha(\cdot)$ as defined in~\autoref{def:capformula}.
    Assume that $\mathcal{A}(I_N) \neq\varnothing$ and $0 \not\in \mathcal{A}(I_N)$\footnote{If $0 \in \mathcal{A}(I_N)$, then we consider instance-based capacity to be $\infty$, since $\mathcal{A}(\Pi^{(N_\mathrm{proj})})$ will always contain $0$ and be nonempty.}. For every $\varepsilon\geq0$,  $\eta > 0$, and $0\leq N_\proj\leq N$, the following hold:
    \begin{enumerate}[label=(\roman*)]
        \item If $\frac{N_\mathrm{proj}}{N} \leq \frac{N_\ib(\varepsilon)}{N} - \beta(\eta)$, then $p(N_\mathrm{proj}) \leq \eta$;
        \item If $\frac{N_\mathrm{proj}}{N} \geq \frac{N_\ib(\varepsilon)}{N} + \beta(\eta)$, then $p(N_\mathrm{proj}) \geq 1 - \eta$,
    \end{enumerate}
    where $\beta(\eta) = \sqrt{8\log(4/\eta)/N}$.
    Note that for any fixed $\eta > 0$, the error term $\beta(\eta)$ goes to zero as $N \to \infty$. Loosely speaking, this means there is a sharp phase transition from $p(N_\mathrm{proj}) \approx 0$ for $N_\mathrm{proj} < N_\ib(\varepsilon)$ to $p(N_\mathrm{proj}) \approx 1$ for $N_\mathrm{proj} > N_\ib(\varepsilon)$ as $N \to \infty$. In particular, this implies
    \[
    |\alpha_\ib^{-1}(\varepsilon) - \alpha^{-1}_\mathrm{simulation}(\varepsilon)| = \frac{1}{P}|N_\ib(\varepsilon)-N_\crit(\varepsilon)| \leq \frac{N}{P} \beta\left(\frac{1}{3}\right),
    \]
    so in the proportional limit with $N \to \infty$ and $N/P = O(1)$, the two capacity estimators $\alpha_\ib(\varepsilon)$ and $\alpha_\mathrm{simulation}(\varepsilon)$ agree. SM Figure \ref{fig:ddcap} gives some numerical intuition for this.
\end{theorem}

\noindent
We prove this result in ~\autoref{sec:ddcap proof details}.
Our proof is based on an application of the \emph{approximate conic kinematic formula} first presented as Theorem I in \cite{Amelunxen2014}. This Theorem bounds the intersection probability of two randomly-oriented convex cones in terms of their \emph{statistical dimensions}, which will almost immediately give rise to the above capacity formula. For our purposes, the two convex cones are, first, the set of admissible regression weights to the original problem in $\mathbb{R}^N$ and, second, a fixed subspace of dimension $\mathbb{R}^{N_\mathrm{proj}}$. This intersection being non-empty will correspond to when the data projected to $\mathbb{R}^{N_\mathrm{proj}}$ remain regressible.

\subsection{Details for the convergence of capacity formula}\label{sec:ddcap proof details}
In this subsection, we prove the concentration of critical dimension (and hence capacity) by a large deviation theory in convex geometry. The definition of statistical dimension and the key theorem are paraphrased below (we specialize Theorem I of \cite{Amelunxen2014} to the case where one of the convex cones is a linear subspace):
\begin{definition}
    We call a subset $C \subseteq \mathbb{R}^N$ a \emph{convex cone} if $C$ is convex and positively homogeneous (i.e., $\lambda \mathbf{x} \in C$ for all $\mathbf{x} \in C$, $\lambda > 0$). Given any closed convex cone $C$, we define it's \emph{statistical dimension} $\delta(C)$ to be
    $$\delta(C) = \mathbb{E}_{\mathbf{t}} \|\Pi_C(\mathbf{t})\|^2,$$
    where $\mathbf{t} \sim \mathcal{N}(0, I_N)$ is a standard Gaussian vector and $\Pi_C(\mathbf{t}) = \arg\min_{\mathbf{x} \in C} \|\mathbf{t} - \mathbf{x}\|$ is the projection of $\mathbf{t}$ onto the convex cone $C$.
\end{definition}
As a consequence of some basic facts about the statistical dimension given in Proposition 3.1 of \cite{Amelunxen2014}, we can equivalently write $\delta(C) = N - \mathbb{E}_{\mathbf{t}} \min_{\mathbf{x} \in C}\|\mathbf{t} - \mathbf{x}\|^2$, which we will take as the definition of the statistical dimension in our proof of Theorem 1. We will also use the fact from the same Proposition 3.1 that statistical dimension agrees with linear-algebraic dimension on subspaces.

Next, we restate the approximate conic kinematic formula—which establishes a large deviation theory for the intersection of two randomly-oriented convex cones—from~\cite{Amelunxen2014}.
\begin{lemma}[Approximate conic kinematic formula, {\cite[Theorem~I]{Amelunxen2014}}]\label{kinematic}
    Let $C$ be a convex cone in $\mathbb{R}^N$ and let $D$ be an $N_\mathrm{proj}$-dimensional linear subspace of $\mathbb{R}^N$. Also let $Q \in \mathbb{R}^{N \times N}$ be a rotation matrix sampled from the orthogonal group $O(N)$ equipped with the normalized Haar measure. For every tolerance parameter $\eta \in (0, 1)$, The following then hold:
    \begin{enumerate}[label=(\roman*)]
        \item If $\delta(C) + N_\mathrm{proj} \leq N - \sqrt{8\log(4/\eta)N}$, then $\mathbb{P}(C \cap QD \neq \{0\}) \leq \eta$;
        \item If $\delta(C) + N_\mathrm{proj} \geq N + \sqrt{8\log(4/\eta)N}$, then $\mathbb{P}(C \cap QD \neq \{0\}) \geq 1 - \eta$.
    \end{enumerate}
\end{lemma}

\noindent
Now let us now return to the proof of our main theorem (\autoref{ddcap}):

\begin{proof}[Proof of Theorem \ref{ddcap}.]
    Recall that in~\autoref{def:ddcap} we define $\cA=\cA(I_N)$ to be the set of admissible regression weights, and the $p(N_\proj)=\mathbb{P}(\cA(\Pi^{(N_\proj)})\neq \varnothing)$ is the probability of being linearly regressible after random projection to $\Real^{N_\proj}$.

    First, let us rewrite $p(N_\proj)$ in terms of the the probability of two randomly-oriented cones intersecting with each other.
    \begin{lemma}\label{lem:prob cone}
    For every $0\leq N_\proj\leq N$, let $D$ be a $N_\proj$-dimensional linear subspace in $\Real^N$. We have
    $$p(N_\proj)=\mathbb{P}(\mathcal{A}(\Pi^{(N_\mathrm{proj})}) \neq \varnothing) = \mathbb{P}(\mathcal{A}^+ \cap QD \neq \{0\}).$$
    \end{lemma}
    \begin{proof}[Proof of~\autoref{lem:prob cone}]
    For fixed $N_\mathrm{proj}$ and $\Pi^{(N_\mathrm{proj})}$, recall that the condition
    \[
    |((\Pi^{(N_\mathrm{proj})})^\top \bm w) \cdot \mathbf{x}^\mu - y^\mu| \leq \varepsilon \textrm{ for all } \mu \in \{1, \dots, P\} \textrm{ and } \mathbf{x}^\mu \in \mathcal{M}^\mu \, ,
    \]
    is used to define membership in $\mathcal{A}(\Pi^{(N_\mathrm{proj})})$. This is in turn equivalent to $(\Pi^{(N_\mathrm{proj})})^\top \bm w \in \mathcal{A}$. Thus, $\mathcal{A}(\Pi^{(N_\mathrm{proj})}) \neq \varnothing$ is equivalent to the condition $\mathcal{A} \cap \operatorname{im} (\Pi^{(N_\mathrm{proj})})^\top \neq \varnothing$. The subspace $\operatorname{im}(\Pi^{(N_\mathrm{proj})})^\top$ is the span of the $N_\mathrm{proj}$ rows of $\Pi^{(N_\mathrm{proj})}$, which by definition is distributed as the span of the first $N_\mathrm{proj}$ rows of an orthogonal matrix $Q$ sampled from $O(N)$ equipped with the normalized Haar measure. This is equivalent to the span of the first $N_\mathrm{proj}$ \emph{columns} of $Q$, which we can write as $QD$ where $D \subseteq \mathbb{R}^N$ is the subspace spanned by the first $N_\mathrm{proj}$ standard basis vectors of $\mathbb{R}^N$ (and in fact be replaced with any arbitrary fixed $N_\mathrm{proj}$-dimensional subspace). Thus, $\mathbb{P}(\mathcal{A}(\Pi^{(N_\mathrm{proj})}) \neq \varnothing) = \mathbb{P}(\mathcal{A} \cap QD \neq \varnothing)$. Since $0 \not\in \mathcal{A}$ and $QD$ is a subspace, the condition $\mathcal{A} \cap QD \neq \varnothing$ is equivalent to $\mathcal{A}^+ \cap QD \neq \{0\}$,
    where $\mathcal{A}^+$ is the cone generated by $\mathcal{A}$. It follows that $\mathbb{P}(\mathcal{A}(\Pi^{(N_\mathrm{proj})}) \neq \varnothing) = \mathbb{P}(\mathcal{A}^+ \cap QD \neq \{0\})$ as desired.
    \end{proof}

    Now back to the proof for~\autoref{ddcap}. 
    Since the set $\mathcal{A}$ is convex, $\mathcal{A}^+$ is a convex cone, so we can apply Theorem \ref{kinematic} with $C = \mathcal{A}^+$ to get:
    \begin{enumerate}[label=(\roman*)]
        \item If $\delta(\mathcal{A}^+) + N_\mathrm{proj} \leq N - \sqrt{8\log(4/\eta)N}$, then $\mathbb{P}(\mathcal{A}^+ \cap QD \neq \{0\}) \leq \eta$;
        \item If $\delta(\mathcal{A}^+) + N_\mathrm{proj} \geq N + \sqrt{8\log(4/\eta)N}$, then $\mathbb{P}(\mathcal{A}^+ \cap QD \neq \{0\}) \geq 1 - \eta$.
    \end{enumerate}
    By definition, we have
    $$\delta(\mathcal{A}^+) = N - \mathbb{E}_\mathbf{t} \min_{\mathbf{w}^+ \in \mathcal{A}^+} \|\mathbf{t} - \mathbf{w}^+\|^2 = N - N_\ib(\varepsilon)$$
    where $N_\ib(\varepsilon)$ is our formula for critical projection dimensionality as defined in~\autoref{def:capformula}. If we plug this expression for $\delta(\mathcal{A}^+)$ and~\autoref{lem:prob cone} into the conditions we got from the approximate conic kinematic formula and rearrange it slightly, we obtain the desired conditions:
    \begin{enumerate}[label=(\roman*)]
        \item If $\frac{N_\mathrm{proj}}{N} \leq \frac{N_\ib(\varepsilon)}{N} - \sqrt{\frac{8\log(4/\eta)}{N}}$, then $\mathbb{P}_{\Pi^{(N_\mathrm{proj})}}(\mathcal{A}(\Pi^{(N_\mathrm{proj})}) \neq \varnothing) \leq \eta$;
        \item If $\frac{N_\mathrm{proj}}{N} \geq \frac{N_\ib(\varepsilon)}{N} + \sqrt{\frac{8\log(4/\eta)}{N}}$, then $\mathbb{P}_{\Pi^{(N_\mathrm{proj})}}(\mathcal{A}(\Pi^{(N_\mathrm{proj})}) \neq \varnothing) \geq 1 - \eta$.
    \end{enumerate}
    
    Finally, we can use this result to compare our direct estimator $\alpha(\varepsilon)$ to $\alpha_\mathrm{simulation}(\varepsilon)$ from~\autoref{def:ddcap}. If $N_\mathrm{proj} \leq N_\mathrm{proj}^* - N\beta(1/3)$, then $p(N_\mathrm{proj}) \leq 1/3$, so $N_\mathrm{proj} \leq N_\crit(\varepsilon)$. If $N_\mathrm{proj} \geq N_\mathrm{proj}^* + N\beta(1/3)$, then $p(N_\mathrm{proj}) \geq 2/3$, so $N_\mathrm{proj} > N_\crit(\varepsilon)$. Combining these, $|N_\crit(\varepsilon) - N_\ib(\varepsilon)| \leq N\beta(1/3)$. Noting that $N_\ib(\varepsilon) = P/\alpha(\varepsilon)$ and dividing the previous inequality by $N$, we get
    \[
    |\alpha_\ib^{-1}(\varepsilon) - \alpha^{-1}_\mathrm{simulation}(\varepsilon)|=|N_\ib(\varepsilon)-N_\crit(\varepsilon)| \leq \frac{N}{P}\beta\left(\frac{1}{3}\right),
    \]
    as desired.
\end{proof}

\noindent
Note that the above instance-based capacity assumes \emph{fixed} target labels. However, in the next section we will demonstrate how Theorem \ref{ddcap} can be used to compute capacity for a model with random labels as well.

\section{Mean-Field Regression Capacity Theory for General Manifolds with Second-Order Correlations}\label{sec:si mean-field}
In this section, we provide details for the mean-field regression capacity using a generative model of data manifolds with second-order correlations and random labels as introduced in~\autoref{sec:mean-field capacity} in the main text. Our mean-field model is an extension from Gardner's formulation of perceptron capacity for point-like neural representations~\cite{gardner1988space} and the recent manifold capacity theory for classification task~\cite{chung2018classification,wakhloo2023linear,GLUE}. The main purpose of mean-field theory in this paper is to provide a framework to derive closed-form formulas for the capacity of synthetic manifold models, such as point-like (see~\autoref{sec:si point-like manifolds}) and sphere-like (see~\autoref{sec:si sphere-like manifolds}) manifolds.

\subsection{Definitions for a mean-field manifold model}\label{sec:si mean-field model}
Fix some $N,P\in\mathbb{N}$, which will be sent to infinity with constant ratio (i.e., $P/N=O(1)$).
We consider $D$-dimensional manifolds of a given (WLOG, convex) shape $\mathcal{S} \subseteq \mathbb{R}^{D+1}$, where we denote the components of $\mathbf{s} \in \mathcal{S}$ by $(s_0, s_1, \dots, s_D)$ and assume that $s_0 = 1$ for all $\mathbf{s} \in \mathcal{S}$. To sample $P$ manifolds embedded in $\mathbb{R}^N$, we will sample an embedding matrix $Q^\mu \in \mathbb{R}^{N \times (D+1)}$ for each manifold index $\mu \in \{1, \dots, P\}$ and then define the $\mu$-th manifold to be $\mathcal{M}^\mu = Q^\mu \mathcal{S} \subseteq \mathbb{R}^N$. Denote the columns of $Q^\mu$ by 
\[
\begin{pmatrix}\mathbf{q}_0 & \mathbf{q}_1 & \cdots & \mathbf{q}_D\end{pmatrix} \, .
\]
Our assumption about the form of $\mathcal{S}$ leads to a natural center-axis decomposition where we think of $r \mathbf{q}_0$ as the embedded manifold's centroid and
\[
\left\{\sum_{i = 1}^D s_i \mathbf{q}_i : \mathbf{s} \in \mathcal{S}\right\}
\]
as the variation around that centroid, with $\{\mathbf{q}_i\, :\, i > 0\}$ being the $D$ axes on which that variation lives. To model correlations between the manifolds, we define the matrices $Q^\mu$ to be zero-mean Gaussian matrices with some covariance structure
\[
\mathbb{E} Q^\mu_{ki}Q^\nu_{lj} = \frac{\delta_{kl}}{N} \Sigma_{\mu i}^{\nu j} \, ,
\]
determined by some positive-definite covariance tensor $\Sigma \in \mathbb{R}^{P(D+1) \times P(D+1)}$. This leads to
\[
\mathbf{q}^\mu_i \cdot \mathbf{q}^\nu_j \overset{\mathcal{P}}\to \Sigma_{\mu i}^{\nu j} \textrm{ as } N \to \infty \, ,
\]
where $\overset{\mathcal{P}}\to$ denotes convergence in probability. The entries $\Sigma_{\mu 0}^{\nu 0}$ of this tensor determine the \enquote{center-center} correlations, the entries $\Sigma_{\mu i}^{\nu j}$ with $i, j > 0$ determine the \enquote{axis-axis} correlations, and finally the entries $\Sigma_{\mu 0}^{\nu i}$ with $i > 0$ determine the \enquote{center-axis} correlations. All of this follows \cite{wakhloo2023linear}, which studied the capacity of essentially the same data model for classification. 

Finally, to define the regression capacity for our mean-field manifold model, we need to analyze the probability of $(\{\cM^\mu\}, \by)$ being linearly regressible. Concretely, for every $\varepsilon>0$, consider
\[
p_{N,P}(\varepsilon,\Sigma,\by) = \mathbb{P}_{Q}(\exists\bw\in\Real^N,\forall \mu,\forall \bx^\mu \in \cM^\mu,\ |\bw\cdot\bx^\mu-y^\mu|\leq\varepsilon) \, ,
\]
where the manifolds $\cM^\mu$ are randomly sampled by the above procedure. We will show that $p_{N,P}$ displays a sharp phase transition indexed by the load $\alpha = P/N$ in the proportional limit ($P,N\to\infty$ with $\alpha$ constant).

To start, we will study the capacity of this model given fixed labels $\{y^\mu\}_{\mu=1}^P$, and afterwards we will generalize to random labels. 

\subsection{Formula for fixed-label mean-field manifolds}\label{sec:SI formula mean-field}
The fixed-label theorem contains the essential ideas of how to pass from the instance-based capacity estimator in \autoref{ddcap} to this more standard mean-field estimator:

\begin{theorem}[Formula for fixed-label mean-field manifolds]\label{corrcap}
    Let a manifold shape $\mathcal{S}$ be given as in the previous section and fix an error tolerance $\varepsilon > 0$. Suppose that for each manifold count $P$, we have a covariance tensor $\Sigma(P) \in \mathbb{R}^{P(D+1) \times P(D+1)}$ and a set of manifold labels $\by(P) \in \mathbb{R}^P$. Assuming that the below limit exists, the capacity $\alpha_\mf(\varepsilon)$ of this model is then given by
    \begin{equation}\label{eq:capacity formula mean-field}
    \alpha_\mf(\varepsilon)^{-1} = \lim_{P \to \infty} \frac{1}{P} \mathbb{E}_{\mathbf{t}} \left[\min_{\bw^+ \in \cA_\mf^+(\epsilon,\by)} \|\mathbf{t} - \bw^+\|^2\right] \, ,
    \end{equation}
    where $\mathbf{t} \sim \mathcal{N}(0, I_{P(D+1)})$ is a standard Gaussian vector, and
    $$
    \cA_\mf(\epsilon,\by) = \left\{\bw \in \mathbb{R}^{P(D+1)} : \left|\bw\cdot((\Sigma(P)^{1/2})^\mu\bs) - y(P)^\mu\right| \leq \varepsilon \textrm{ for all } \mu \in \{1, \dots, P\} \textrm{ and } \mathbf{s} \in \mathcal{S}\right\} \, ,
    $$
    where $(\Sigma(P)^{1/2})^\mu$ denotes the $\mu$-th column block of $\Sigma(P)^{1/2}$, and as before $\cA_\mf^+(\epsilon,\by)$ denotes the cone generated by $\cA_\mf(\epsilon,\by)$.
    
    To be more concrete, let $p_{N,P}(\varepsilon,\Sigma(P), \by(P))$ be the probability of an admissible readout $\bw$ existing when we sample $P$ manifolds in $\mathbb{R}^N$ from our mean-field model, as in the previous section. If we then send $P, N \to \infty$ with constant load $\alpha = P/N < \alpha_\mf(\varepsilon)$, we get $p_{N,P}(\varepsilon) \to 1$. If instead $\alpha > \alpha_\mf(\varepsilon)$, we get $p_{N,P}(\varepsilon) \to 0$.
\end{theorem}
\begin{proof}[Proof of~\autoref{corrcap}]
    The overall idea of the proof is to rewrite the mean-field manifold model to a random projection problem so that we can invoke~\autoref{ddcap}.

    \vspace{3mm}
    \noindent\textbf{Step 1: Rewrite the mean-field manifold model.}
    Given a mean-field manifold model (as described in~\autoref{sec:si mean-field model}) $\{\cM^\mu\}$ where $\cM^\mu=Q^\mu\cS\subset\Real^N$ with shape $\cS$, correlation tensor $\Sigma$, disorder $\{Q^\mu\}$, and parameters $N,P$, and a fixed label vector $\by\in\Real^P$. Let $\{\overline{\cM}^\mu\subset\Real^{P(D+1)}\}$ be the decorrelated version of manifolds specified by $\overline{\cM}^\mu = (\Sigma^{1/2})^\mu\cS$ where $(\Sigma^{1/2})^\mu$ denotes the $\mu$-th column block of $\Sigma^{1/2}$, and $G\in\Real^{N\times P(D+1)}$ be the decorrelated version of disorder specified by $G=(\Sigma^{-1/2}Q)^\top\in\Real^{N\times P(D+1)}$. Note that now $\{\overline{\cM}^\mu\}$ is deterministic and $G$ has entries drawn i.i.d.~from $\mathcal{N}(0, 1/N)$.
    
    Next, observe that $\cM^\mu = G^\mu \overline{\cM}^\mu$. We can therefore rewrite the condition
    \[ \forall \mu, \forall\bx^\mu\in\cM^\mu,
    |\bw\cdot\bx^\mu-y^\mu|\leq\varepsilon
    \]
    from the definition of $p_{N,P}(\varepsilon,\by)$ as
     \[ \forall \mu, \forall \overline{\bx}^\mu \in \overline{\cM}^\mu, 
    |\bw\cdot G\overline{\bx}^\mu-y^\mu|\leq\varepsilon.
    \]
    To connect this to our notion of instance-based capacity, we use the standard fact that $G \in \mathbb{R}^{N \times P(D+1)}$, with entries drawn i.i.d. from $\mathcal{N}(0, 1/N)$, is equivalent to a random projection from $\mathbb{R}^{P(D+1)}$ to $\mathbb{R}^N$ as $N \to \infty$. Thus, $G$ corresponds to $\Pi^{(N_\proj)}$ from~\autoref{ddcap} and we can think of the mean-field manifolds $\{\cM^\mu\}$ as random projections of the deterministic manifolds $\{\overline{\cM}^\mu\}$ (i.e., $\cM^\mu=G^\mu\overline{\cM}^\mu$ where $G^\mu$ is the $\mu$-th column block of $G$).
    
    \vspace{3mm}
    \noindent\textbf{Step 2: Invoke~\autoref{ddcap} for random projection model.}
    Now that we have established an equivalent random projection model for the mean-field manifold model, we are able to invoke~\autoref{ddcap} and show the concentration of linear regressibility of the latter. Concretely, for every $\eta>0$, we have
    \begin{enumerate}[label=(\roman*)]
        \item If $\frac{N}{P(D+1)} \leq \frac{N_\mf(\varepsilon)}{P(D+1)} - \sqrt{\frac{8\log(4/\eta)}{P(D+1)}}$, then $\mathbb{P}_{G}(\mathcal{A}(G) \neq \varnothing) \leq \eta$;
        \item If $\frac{N}{P(D+1)} \geq \frac{N_\mf(\varepsilon)}{P(D+1)} + \sqrt{\frac{8\log(4/\eta)}{P(D+1)}}$, then $\mathbb{P}_{G}(\mathcal{A}(G) \neq \varnothing) \geq 1 - \eta$,
    \end{enumerate}
    where
    $$
    N_\mf(\varepsilon) = \mathbb{E}_\mathbf{t}\left[\min_{\bw^+ \in \mathcal{A}^+} \|\bm{t} - \bw^+\|^2\right]
    $$
    and
    \[
    \cA(M) = \left\{\bw\in\Real^{P(D+1)}\, :\, |\bw\cdot(M\overline{\bx}^\mu)-y^\mu|\leq\varepsilon ,\ \forall\overline{\bx}^\mu\in\overline{\cM}^\mu,\ \forall\mu\right\} \, .
    \]
    Finally, notice that $\cA_\mf=\cA(I_{P(D+1)})$, $\alpha_\mf^{-1}(\varepsilon)=N_\mf(\varepsilon)/P$, and $p_{N,P}(\varepsilon,\by)=\mathbb{P}_{G}(\mathcal{A}(G))$, this completes the proof for~\autoref{corrcap}.
\end{proof}

\subsection{Formula for random-label mean-field manifolds}\label{sec:SI random label}
In~\autoref{sec:SI formula mean-field}, we derived a formula (\autoref{eq:capacity formula mean-field}) for the manifold capacity of our mean-field model with \emph{fixed} labels. However, the toy models we study in later sections also include randomness in the labels. To incorporate this extra disorder, we consider the probability
\[
p_{N,P}(\varepsilon, \Sigma) = \mathbb{P}_{Q, \by}\left(\exists \bw \in \mathbb{R}^N, \forall \mu, \forall \bx^\mu \in \cM^\mu, |\bw \cdot \bx^\mu - y^\mu| \leq \varepsilon\right) \, ,
\]
of an admissible readout $\bm w$ existing for our model, where this probability is now also evaluated over some distribution of labels. We can write $p_{N,P}(\varepsilon, \Sigma) = \mathbb{E}_\by p_{N,P}(\varepsilon, \Sigma, \by)$ in terms of the probability of an admissible readout existing for our fixed-label model. To show that $p_{N,P}(\varepsilon, \Sigma)$ exhibits a sharp phase transition at some capacity $\alpha = P/N$ in the proportional limit $P,N \to \infty$, it is sufficient to show that the label-conditioned model $p_{N,P}(\varepsilon, \Sigma, \by)$ exhibits a phase transition at the same $\alpha$ for almost all labels $\by$.

To prove this, it in turn suffices to show that the label-conditioned capacity estimator 
\[
\alpha_\mf(\varepsilon, P, \by)^{-1} = \frac{1}{P} \mathbb{E}_\bt \left[\min_{\bw^+ \in \mathcal{A}^+} \|\bm{t} - \bw^+\|^2\right]
\]
from \autoref{corrcap} concentrates around its mean $\mathbb{E}_\by \alpha_\mf(\varepsilon, P, \by)^{-1}$ as $P \to \infty$. If this concentration holds, the capacity of the mean-field model \emph{with random labels} is then well-defined (i.e., a sharp phase transition exists) and is given by
\begin{equation}\label{eq:capacity formula random label mean-field}
\alpha^{-1}_\mf(\varepsilon) = \lim_{P \to \infty} \mathbb{E}_\by [\alpha_\mf(\varepsilon, P, \bm y)^{-1}] \, .
\end{equation}
We will show that this concentration holds for the toy models studied in the following sections.

\section{Regression Capacity of Random Point-Like Manifolds}\label{sec:si point-like manifolds}
Our mean-field theory for regression capacity (\autoref{sec:si mean-field}) enables analytical study of synthetic manifolds. In this section, we focus on deriving closed-form formulas for the regression capacity of random point-like manifolds with different correlation structures. Later in~\autoref{sec:SI corrsphere}, we will study the regression capacity of a data model with spherical manifolds which are correlated in multiple ways.

\subsection{Capacity of Uncorrelated Points}\label{sec:SI uncor points}
We start with studying the capacity of uncorrelated random \emph{points} with uncorrelated target labels.

\subsubsection{Our toy model for uncorrelated points and its regression capacity}
To sample $P$ points $\{\bm{x}^\mu \in \mathbb{R}^N\}_{\mu=1}^P$ embedded in $\mathbb{R}^N$ with expected square-norm $\mathbb{E} \|\bm{x}^\mu\|^2 = r^2$ set by some $r > 0$, we define
\[
x^{\mu}_{k} \overset{i.i.d.}\sim \mathcal{N}\left(0, \frac{r^2}{N}\right) \textrm{ for all } \mu \in \{1, \dots, P\} \textrm{ and } k \in \{1, \dots, N\} \, .
\]
This scaling ensures that $\mathbb{E} \bm{x}^{\mu} \cdot \bm{x}^\nu = r^2\delta_{\mu \nu}$. We sample the target labels $\{y_\mu\}_{\mu=1}^P$ for each point i.i.d. from $\mathcal{N}(0, \sigma^2)$ with some label variance $\sigma^2 > 0$. In the notation of~\autoref{sec:si mean-field model}, we have manifold shape $\mathcal{S} = \{1\} \subseteq \mathbb{R}^{D+1}$ with dimension $D=0$, a correlation tensor $\Sigma \in \mathbb{R}^{P(D+1) \times P(D+1)}$ with entries $\Sigma_{\mu 0}^{\nu 0} = r^2\delta_{\mu \nu}$, and label correlation $H \in \mathbb{R}^{P \times P}$ with entries $H_{\mu \nu} = \sigma^2\delta_{\mu \nu}$.

Here, we use the formula for random-label mean-field manifolds (\autoref{eq:capacity formula random label mean-field}) to derive a simple capacity formula for the toy model described above:
\begin{theorem}[Formula for uncorrelated points]\label{thm:uncorelated points}
The uncorrelated points model described above has capacity given by
\begin{equation}\label{eq:SI uncorr-points-cap}
\alpha^{-1}(\varepsilon, r, \sigma) = \min_{\tau \geq 0} 2\int_{b / a}^\infty \mathcal{D}z (a z - b)^2 \, ,
\end{equation}
where $\mathcal{D}z$ denotes the standard Gaussian measure and
$a = \sqrt{1+\frac{\sigma^2\tau^2}{r^2}}$, $b = \frac{\varepsilon\tau}{r}$.
\end{theorem}
We interpret this formula and numerically validate it in the main text. The derivation details for~\autoref{thm:uncorelated points} is provided in the following subsubsection for the completeness of presentation.

\subsubsection{Derivation details for~\autoref{thm:uncorelated points}}

As described in \autoref{sec:SI random label}, we will first study the label-conditioned capacity estimator
\[
\alpha_\mf(\varepsilon, P, \bm y)^{-1} = \frac{1}{P} \mathbb{E}_\bt \left[\min_{\bw^+ \in \cA^+} \|\bt - \bw^+\|^2\right] \, ,
\]
which we can rewrite as $\alpha_\mf(\varepsilon, P, \bm y)^{-1} = \mathbb{E}_\bt \min_{\tau \geq 0} f(\bt, \by, \tau)$ in terms of the field $f(\bt, \by, \tau) = \frac{1}{P}\min_{\bw \in \cA(\by)} \|\bt - \tau \bw\|^2$ and constraint set $\mathcal{A}(\by) = \{\bw \in \mathbb{R}^P : |r w^\mu - y^\mu| \leq \varepsilon\}$.
For the uncorrelated case, this field factorizes into $P$ terms:
\[
f(\bt, \by, \tau) = \frac{1}{P} \sum_{\mu=1}^P \min_{\substack{w^\mu \in \mathbb{R} \\ |r w^\mu - y^\mu| \leq \varepsilon}} (t^\mu - \tau w^\mu)^2 \, .
\]
We can make a change of variables to $h^\mu = rw^\mu - y^\mu$ and write
\[
f(\bt, \by, \tau) = \frac{1}{P} \sum_{\mu=1}^P \min_{\substack{h^\mu \in \mathbb{R} \\ |h^\mu| \leq \varepsilon}} \left(\left(t^\mu -\frac{\tau}{r} y^\mu\right) - \frac{\tau}{r} h^\mu\right)^2 \, .
\]
Observe that $f(\bt, \by, \tau)$ is identically distributed to
\[
g(\bz, \tau) = \frac{1}{P} \sum_{\mu=1}^P \min_{\substack{h^\mu \in \mathbb{R} \\ |h^\mu| \leq \varepsilon}} \left(a z^\mu - \frac{\tau}{r} h^\mu\right)^2
\]
with $\bz \sim \mathcal{N}(0, I_P)$ and $a = \sqrt{1 + \frac{\sigma^2\tau^2}{r^2}}$. Each term of the above sum is now i.i.d., so by the law of large numbers $g(\bz, \tau)$ concentrates to the deterministic limit
\[
\tilde g(\tau) = \mathbb{E}_{z \sim \mathcal{N}(0, 1)} \min_{|h| \leq \varepsilon} \left(az - \frac{\tau}{r} h\right)^2 = 2 \int_{b/a}^\infty \mathcal{D} z (a z - b)^2 \, ,
\]
as $P \to \infty$, where $b = \frac{\varepsilon \tau}{r}$. Thus, $f(\bt, \by, \tau)$ also concentrates around $\tilde g(\tau)$ as $P \to \infty$. It follows that the label-conditioned capacity $\alpha_\mf(\varepsilon, P, \bm y)^{-1} = \mathbb{E}_\bt \min_{\tau \geq 0} f(\bt, \by, \tau)$ concentrates around $\min_{\tau \geq 0} \tilde g(\tau)$ as $P \to \infty$, which is therefore the capacity of this uncorrelated points model as argued in~\autoref{sec:SI random label}, which is what we wanted to show in~\autoref{thm:uncorelated points}.

\subsection{Capacity of Correlated Points}\label{sec:SI correlated points}
Next, we introduce uniform correlations between the points $\{\bm{x}^\mu \in \mathbb{R}^N\}_{\mu=1}^P$, controlled by a parameter $\psi \in [0, 1)$, and the labels $\{y^\mu \in \mathbb{R}\}_{\mu=1}^P$, controlled by a parameter $\rho \in [0, 1)$. We show in this subsection that the capacity of this model reduces to the capacity of the uncorrelated points model with rescaled parameters. 

\subsubsection{Our toy model for correlated points and its regression capacity}
In the notation of~\autoref{sec:si mean-field}, we define the data correlation tensor by $\Sigma_{\mu 0}^{\nu 0} = r^2[(1-\psi)\delta_{\mu \nu} + \psi]$ and the label correlation matrix by $H_{\mu\nu} = \sigma^2 [(1-\rho)\delta_{\mu\nu}+\rho]$.

Here, we derive a simple capacity formula for the toy model described above by reducing the capacity formula for a correlated point model to an equivalent uncorrelated point model.
\begin{theorem}[Formula for correlated points]\label{thm:corelated points}
Given the model for correlated points described above. Let $\alpha^{-1}(\varepsilon, r,\sigma,\rho,\psi,P)$ be the capacity formula from~\autoref{corrcap} for this model, we have that $\alpha^{-1}(\varepsilon,r, \sigma,\rho,\psi, P)\xrightarrow{P \to \infty}\alpha^{-1}(\varepsilon,r, \sigma,\rho,\psi)$ with
\begin{equation}\label{eq:SI corr-points-cap}
\alpha^{-1}(\varepsilon,r, \sigma,\rho,\psi) = \alpha^{-1}(\varepsilon,r\sqrt{1-\psi},\sigma\sqrt{1-\rho}) = 2\int_{b / a}^\infty \mathcal{D}z (a z - b)^2
\end{equation}
where $\mathcal{D}z$ denotes the standard Gaussian measure and
$a = \sqrt{1+\frac{(1-\rho)\sigma^2\tau^2}{(1-\psi)r^2}}$, $b = \frac{\varepsilon\tau}{(1-\psi)r}$.
\end{theorem}
We interpret this formula and numerically validate it in the main text. The derivation details for~\autoref{thm:corelated points} is provided in the following subsubsection for the completeness of presentation.

\subsubsection{Derivation details for~\autoref{thm:corelated points}}
For a given label tolerance $\varepsilon > 0$,~\autoref{corrcap} gives us a capacity estimator
\[
\alpha^{-1}(\varepsilon) = \min_{\tau \geq 0} \frac{1}{P} \mathbb{E}_{\bm{t},\bm{y}} \min_{\bm{u} \in \mathcal{B}(\bm{y})} \left\Vert \bm{t} - \tau \Sigma^{-\frac{1}{2}} \bm{u} \right\Vert^2 \, ,
\]
where
\[
\mathcal{B}(\bm{y}) = \left\{\bm{u} \in \mathbb{R}^P : |u_\mu - y_\mu| \leq \varepsilon \textrm{ for all } \mu \in \{1, \dots, P\}\right\} \, .
\]
As in the previous section, our goal is to compute
$\lim_{P \to \infty} \alpha^{-1}(\varepsilon, P).$
Similarly to the previous section, we can define a shifted variable $\bm{v} = \bm{u} - \bm{y}$ so that our estimator becomes
\begin{align*}
    \alpha^{-1}(\varepsilon, P) &= \frac{1}{P} \min_{\tau \geq 0} \mathbb{E}_{\bm{t},\bm{y}} \min_{\bm{v} \in \mathcal{B}(0)} \left\Vert \bm{t} - \tau\Sigma^{-\frac{1}{2}} (\bm{v}+\bm{y}) \right\Vert^2 \\
    &= \frac{1}{P} \min_{\tau \geq 0} \mathbb{E}_{\bm{t},\bm{y}} \min_{\bm{v} \in \mathcal{B}(0)} \left\Vert \left(\bm{t}-\tau\Sigma^{-\frac{1}{2}} \bm{y}\right) - \tau\Sigma^{-\frac{1}{2}} \bm{v} \right\Vert^2 \\
    &= \frac{1}{P} \min_{\tau \geq 0} \mathbb{E}_{\bm{z}} \min_{\bm{v} \in \mathcal{B}(0)} \left\Vert \Omega^{1/2}\bm{z} - \tau\Sigma^{-\frac{1}{2}} \bm{v} \right\Vert^2, \\
\end{align*}
where $\Omega = I + \tau^2 \Sigma^{-\frac{1}{2}}H\Sigma^{-\frac{1}{2}}$, so that we have combined $\bm{t}$ and $\bm{y}$ into a single random vector.

The key to analyzing this model is to decompose the correlation matrices into their eigencomponents, with one eigenspace spanned by the vector $\mathbf{1}/\sqrt{P}$, where $\bm 1 = (1, 1, \dots, 1)$, and another $(P-1)$-dimensional eigenspace orthogonal to this vector. In particular, we have:
\begin{align*}
    \Sigma &= r^2\left[(1-\psi)I + \psi \mathbf{1}\mathbf{1}^\top\right] \\
    &= r^2(1-\psi) \Pi + r^2(1-\psi+P\psi)\left[\frac{1}{P}\mathbf{1}\mathbf{1}^\top\right],
\end{align*}
and
\begin{align*}
    H &= \sigma^2\left[(1-\rho)I + \rho \mathbf{1}\mathbf{1}^\top\right] \\
    &= \sigma^2(1-\rho) \Pi + \sigma^2(1-\rho+P\rho)\left[\frac{1}{P}\mathbf{1}\mathbf{1}^\top\right],
\end{align*}
where $\Pi = I - \frac{1}{P} \bm 1 \bm 1^\top$. It follows that
\begin{align*}
    \Sigma^{-1/2} &= \frac{1}{r\sqrt{1-\psi}} \Pi + \frac{1}{r\sqrt{1-\psi+P\psi}}\left[\frac{1}{P}\mathbf{1}\mathbf{1}^\top\right],
\end{align*}
and
\begin{align*}
    \Omega^{1/2} &= \sqrt{1 + \tau^2 \frac{\sigma^2(1-\rho)}{r^2(1-\psi)}} \Pi + \sqrt{1 + \tau^2\frac{\sigma^2(1-\rho+P\rho)}{r^2(1-\psi+P\psi)}}\left[\frac{1}{P}\mathbf{1}\mathbf{1}^\top\right].
\end{align*}
We can write our capacity estimator as:
$$\alpha^{-1}(\varepsilon, P) = \min_{\tau \geq 0} \mathbb{E}_{\bm{z}} f_{\tau}(\bm{z}),$$
where we split the field $f_{\tau}(\bm{z})$ into two terms corresponding to the two eigenspaces of the correlation matrices:
\begin{align*}
    f_{\tau}(\bm{z}) = \min_{\bm{v} \in \mathcal{B}(0)} &\frac{1}{P}\left\Vert\sqrt{1+\tau^2\frac{\sigma^2(1-\rho)}{r^2(1-\psi)}} \Pi \bm{z} - \frac{\tau}{r\sqrt{1-\psi}} \Pi \bm{v}\right\Vert^2 \\
    +\frac{1}{P}&\left\Vert\sqrt{1+\tau^2\frac{\sigma^2(1-\rho+P\rho)}{r^2(1-\psi+P\psi)}} \left[\frac{1}{P}\mathbf{1}\mathbf{1}^\top\right]\bm{z} - \frac{\tau}{r\sqrt{1-\psi+P\psi}} \left[\frac{1}{P}\mathbf{1}\mathbf{1}^\top\right]\bm{v}\right\Vert^2.
\end{align*}

Notice that we can rewrite the second term above as
\begin{equation}\label{eq:uncorr_pt_vanishing_term}
    \left(\sqrt{1+\tau^2\frac{\sigma^2(1-\rho+P\rho)}{r^2(1-\psi+P\psi)}} \frac{1}{P} \mathbf{1}^\top \bm{z} - \frac{\tau}{Pr\sqrt{1-\psi+P\psi}} \mathbf{1}^\top \bm{v}^*\right)^2,
\end{equation}
where $\bm{v}^*$ is a minimizing $\bm{v}$. Since $\bm{z}$ is just a standard Gaussian vector, the law of large numbers gives
\[
\frac{1}{P} \mathbf{1}^\top \bm{z} \overset{P}\to 0 \, ,
\]
where $\overset{P}\to$ denotes convergence in probability. Since
\[
\sqrt{1+\tau^2\frac{\sigma^2(1-\rho+P\rho)}{r^2(1-\psi+P\psi)}} = O(1) \, ,
\]
we have
\[
\sqrt{1+\tau^2\frac{\sigma^2(1-\rho+P\rho)}{r^2(1-\psi+P\psi)}} \frac{1}{P}\mathbf{1}^\top \bm{z} \overset{P}\to 0 \, .
\]
As for the second term, we must always have $|\mathbf{1}^\top \bm{v}^*/P| \leq \varepsilon$ because $\bm{v}^* \in \mathcal{B}(0)$. Thus,
\[
\left|\frac{\tau}{Pr\sqrt{1-\psi+P\psi}}\mathbf{1}^\top \bm{v}^*\right| \leq \frac{\varepsilon \tau}{r\sqrt{1-\psi+P\psi}} \to 0 \, ,
\]
so the second term in Eq.~\eqref{eq:uncorr_pt_vanishing_term} and therefore all of Eq.~\eqref{eq:uncorr_pt_vanishing_term} converges in probability to zero, \emph{regardless of what the optimal $\bm{v}^*$ actually is}. For the purpose of computing the asymptotic capacity of our model, we can therefore replace $f_\tau(\bm{z})$ with the simpler field
\[
\tilde f_{\tau}(\bm{z}) = \min_{\bm{v} \in \mathcal{B}(0)} \frac{1}{P}\left\Vert\sqrt{1+\tau^2\frac{\sigma^2(1-\rho)}{r^2(1-\psi)}} \Pi \bm{z} - \frac{\tau}{r\sqrt{1-\psi}} \Pi \bm{v}\right\Vert^2 \, .
\]
Except for the interspersed projection matrices $\Pi$, this looks a lot like the capacity estimator for the \emph{uncorrelated} points model studied in the previous section with rescaled parameters $r_\mathrm{equiv} = r\sqrt{1-\psi}$ and $\sigma_\mathrm{equiv} = \sigma\sqrt{1-\rho}$. Indeed, we will show that the asymptotic capacity of our correlated points model is equal to the capacity of this uncorrelated model. The key step is to show that $\bar{v} = \frac{1}{P}\mathbf{1}^\top \bm{v}^*$ converges to zero as $P \to \infty$. This is true for the uncorrelated case ($\psi, \rho = 0$) by the law of large numbers since all the components $\{v_\mu^*\}$ are i.i.d., but it takes a little work to show that this also holds for the correlated case. Since we will reuse the same argument later in \autoref{sec:SI corrsphere}, we abstract it into a proposition that we prove at the end of this section:

\begin{proposition}\label{prop:vbar}
   Let $\bm v^*(\bm z)$ be a minimizer of $\|\Pi (a \bm z - \bm v)\|$ subject to $\bm v \in B$, where $\bm z \in \mathbb{R}^P$ is a standard Gaussian vector and $B = \prod_{\mu=1}^P [-\varepsilon^\mu, \varepsilon^\mu]$ for some sequence $\{\varepsilon^\mu\}_{\mu=1}^P$ satisfying $0 \leq \varepsilon^\mu \leq \varepsilon$ for all $\mu$ and some fixed $\varepsilon \geq 0$. Then
   \[
   \mathbb{P}_z(|\bar v^*(z)| > \eta) \leq \exp(-C\eta^2(1 - \erf(\varepsilon/|a|))^2 P/a^2)
   \]
   where $\bar x := \bm{1}^\top \bm x / P$ and $C > 0$ is a universal constant. In particular, for any fixed $\varepsilon$ and $a$, $\bar v$ converges in probability to $0$ as $P \to \infty$.
\end{proposition}

This proposition can be applied to the minimizer $\bm v = \bm v^*(\bm z)$ of the field $\tilde f_\tau(\bm z)$ to show that $\bar v \to 0$ as $P \to \infty$.
The capacity for the \emph{uncorrelated} points model with $r_\mathrm{equiv} = r\sqrt{1-\psi}$ and $\sigma_\mathrm{equiv} = \sigma\sqrt{1-\rho}$ can be written as
\[
\alpha^{-1}_\mathrm{uncorr}(\varepsilon, P) = \min_{\tau \geq 0} \mathbb{E}_{\bm{z}} g_\tau(\bm{z}) \, ,
\]
where
\[
g_\tau(\bm{z}) = \min_{\bm{v} \in \mathcal{B}(0)} \frac{1}{P} \|a \bm{z} - b \bm{v}\|^2 \, .
\]
We can split this into eigencomponents like before:
\[
g_\tau(\bm{z}) = \min_{\bm{v} \in \mathcal{B}(0)} \frac{1}{P} \|\Pi (a \bm{z} - b \bm{v})\|^2 + (a \bar z - b \bar v)^2 \, .
\]
Since $\bar z \overset{P}\to 0$ as $P \to \infty$, we can equivalently consider the field
\[
\min_{\bm{v} \in \mathcal{B}(0)} \frac{1}{P} \|a \Pi \bm{z} - b \Pi \bm{v}\|^2 + (b\bar v)^2 \, .
\]
The second term is minimized when $\bar v = 0$. As remarked above, we know that minimizing the first term alone (as in the field $\tilde f_\tau(\bm{z})$) leads to $\bar v \overset{P}\to 0$ as $P \to \infty$ anyway. Thus, $g_\tau(\bm{z})$ and $\tilde f_\tau(\bm{z})$ are asymptotically equivalent, so the capacity of our correlated model is equal to the capacity of the uncorrelated points model with rescaled centroid norm $r_\mathrm{equiv} = r\sqrt{1-\psi}$ and label scale $\sigma_\mathrm{equiv} = \sigma\sqrt{1-\rho}$.

We use the following lemma to prove~\autoref{prop:vbar}:

\begin{lemma}[Concentration for Lipschitz functions]\label{lem:concentration Lipschitz}
    Let $f : \mathbb{R}^P \to \mathbb{R}$. Suppose $f(\cdot)$ is $L$-Lipschitz. Then there exists a constant $C>0$ such that
    \[
    \mathbb{P}_{\bt\sim\cN(0,I_P)}(|f(\bt) - \mathbb{E}_{\bt'} f(\bt')| > \Delta) \leq \exp\left(-C\cdot\frac{\Delta^2}{L^2}\right)
    \]
    for every $\Delta > 0$.
\end{lemma}
The proof for the above lemma can be found in standard textbooks such as Theorem V.1 in \cite{Milman2001-yh} and is often attributed to Maurey and Pisier.

\subsubsection{Proof of \autoref{prop:vbar}}

We can study $\bm v^*(\bm z)$ using the Lagrangian
$$\mathcal{L}(\bm v, \bm{\lambda}_+, \bm{\lambda}_-) = \frac{1}{2} \|\Pi (a \bm z - \bm v)\|^2 + \sum_{\mu=1}^P \lambda^\mu_+ (x^\mu - \varepsilon^\mu) + \sum_{\mu=1}^P \lambda^\mu_- (-x^\mu - \varepsilon^\mu).$$
The KKT conditions characterizing $\bm v = \bm v^*(\bm z)$ are:
\begin{itemize}
    \item Stationarity: $\Pi (\bm v - a \bm z) + \bm \lambda = 0$;
    \item Feasibility: $\bm v \in B$;
    \item Complementary slackness: $\lambda^\mu \neq 0 \Rightarrow v^\mu = \operatorname{sgn}(\lambda^\mu) \varepsilon^\mu$,
\end{itemize}
where we define $\bm \lambda = \bm \lambda_+ - \bm \lambda_-$. Recalling that $\Pi \bm x = \bm x - \bar x \bf{1}$, these conditions imply
\[
v^\mu = \max(-\varepsilon^\mu, \min(\varepsilon^\mu, \bar v + a [\Pi \bm z]^\mu)) \textrm{ for all } \mu \, .
\]
This leads to a self-consistent equation for $\bar v$:
\[
\bar v = \frac{1}{P} \sum_{\mu=1}^P \max(-\varepsilon^\mu, \min(\varepsilon^\mu, \bar v + a[\Pi \bz]^\mu)) = \frac{\bm 1^\top}{P} \operatorname{proj}_B(\bar v \bm 1 + a \Pi \bm z) \, .
\]
Define
\[
f(\bar v, \bm z) = \frac{1}{P} \sum_{\mu=1}^P \max(-\varepsilon^\mu, \min(\varepsilon^\mu, \bar v + a[\Pi \bz]^\mu)) = \frac{\bm 1^\top}{P} \operatorname{proj}_B(\bar v \bm 1 + a \Pi \bm z)\]
to be the right-hand size of the self-consistent equation, so that $\bar v = f(\bar v, \bz)$ for the optimal $\bv$. Since $\Pi$ is 1-Lipschitz, the map $\bz \mapsto \bar v \bm 1 + a \Pi \bz$ is $a$-Lipschitz. Projection onto any closed, convex set (here $B$) is 1-Lipschitz, so the composition of these maps $\bz \mapsto \proj_B(\bar v + a \Pi \bz)$ is still $a$-Lipschitz. Since $\|\bm 1/P\| = 1/\sqrt{P}$, it follows that $f(\bar v, \bz)$ is $(a/\sqrt{P})$-Lipschitz in $\bz$. We can now invoke our concentration inequality \autoref{lem:concentration Lipschitz} for Lipschitz functions to obtain
\[
\mathbb{P}\left(|f(\bar v, \bm z) - \mathbb{E}_\bz f(\bar v, \bm z)| > \Delta\right) \leq \exp\left(-C \Delta^2 P/a^2\right)
\]
for some universal constant $C > 0$. This essentially gives us a \emph{deterministic} self-consistent equation $\bar v \approx \mathbb{E}_\bz f(\bar v, \bm z)$ that $\bar v$ must satisfy for the optimal $\bv$, and we will show below that this implies $\bar v$ concentrates at 0 as $P \to \infty$. Specifically, we can write
\[
\mathbb{E}_\bz f(\bar v, \bm z) = \frac{1}{P} \sum_{\mu=1}^P \mathbb{E}_\bz \max(-\varepsilon^\mu, \min(\varepsilon^\mu, \bar v + a[\Pi \bm z]^\mu)) \, .
\]
Noting that $a [\Pi z]^\mu$ is Gaussian with standard deviation $\tilde a = a \sqrt{1-1/P}$, we can write
\[
\mathbb{E}_\bz \max(-\varepsilon^\mu, \min(\varepsilon^\mu, \bar v + a [\Pi z]^\mu)) = \mathbb{E}_{g \sim \mathcal{N}(0, 1)} \max(-\varepsilon^\mu, \min(\varepsilon^\mu, \bar v + \tilde a g)) \, ,
\]
which can be explicitly computed. In particular, one can show that
\[
\frac{\partial}{\partial \varepsilon^\mu} \mathbb{E}_{g \sim \mathcal{N}(0, 1)} \max(-\varepsilon^\mu, \min(\varepsilon^\mu, \bar v + \tilde a g)) = \frac{1}{2} \erf\left(\frac{\bar v + \varepsilon^\mu}{\tilde a\sqrt{2}}\right) + \frac{1}{2} \erf\left(\frac{\bar v - \varepsilon^\mu}{\tilde a \sqrt{2}}\right) \, .
\]
First, consider the case where $\bar v \geq 0$. The above derivative w.r.t. $\varepsilon^\mu$ is positive, and since $\varepsilon^\mu \leq \varepsilon$, we have
\[
0 \leq \mathbb{E} \max(-\varepsilon^\mu, \min(\varepsilon^\mu, \bar v + a [\Pi \bm z]^\mu)) \leq \mathbb{E} \max(-\varepsilon, \min(\varepsilon, \bar v + a [\Pi \bm z]^\mu)) \, .
\]
for all $\mu$. Similarly, when $\bar v \leq 0$, the above derivative w.r.t. $\varepsilon^\mu$ is negative, and we get
\[
\mathbb{E} \max(-\varepsilon, \min(\varepsilon, \bar v + a [\Pi \bm z]^\mu)) \leq \mathbb{E} \max(-\varepsilon^\mu, \min(\varepsilon^\mu, \bar v + a [\Pi \bm z]^\mu)) \leq 0 \, .
\]
It follows that, for any $\bar v$,
$|\mathbb{E} f(\bar v, \bm z)| \leq |\mathbb{E} \tilde f(\bar v, \bm z)|$, where we define
\[
\tilde f(\bar v, \bm z) = \frac{1}{P} \sum_{\mu=1}^P \max(-\varepsilon, \min(\varepsilon, \bar v + a [\Pi \bm z]^\mu)) \, .
\]
Since the terms of this sum are now identically distributed, we get
\[
\mathbb{E} \tilde f(\bar v, \bm z) = \mathbb{E}_{g \sim \mathcal{N}(0, 1)} \max(-\varepsilon, \min(\varepsilon, \bar v + \tilde a g)) \, .
\]
One can compute
\[
\frac{\partial}{\partial \bar v} \mathbb{E} \tilde f(\bar v, \bm z) = -\frac{1}{2} \operatorname{erf}\left(\frac{\bar v-\varepsilon}{|\tilde a|\sqrt{2}}\right) + \frac{1}{2}\operatorname{erf}\left(\frac{\bar v + \varepsilon}{|\tilde a| \sqrt{2}}\right) \, ,
\]
(where the derivative is now w.r.t. $\bar v$), which satisfies
\[
0 \leq \frac{\partial}{\partial \bar v} \mathbb{E} \tilde f(\bar v, \bm z) \leq \erf\left(\frac{\varepsilon}{|\tilde a| \sqrt{2}}\right) < 1 \textrm{ for all } \bar v \, .
\]
Noting also that $\mathbb{E} \tilde f(0, \bz) = 0$ by symmetry, it follows that $|\mathbb{E} f(\bar v, \bm z)| \leq |\mathbb{E} \tilde f(\bar v, \bm z)| \leq |\bar v| \cdot \erf(\varepsilon/(|\tilde a| \sqrt{2}))$. WLOG, we may assume $P \geq 2$ so that $|\tilde a| \geq |a| / \sqrt{2}$, and we get $|\mathbb{E} f(\bar v, \bm z)| \leq |\bar v| \cdot \operatorname{erf}(\varepsilon/|a|) < 1$. Define $\Delta = \eta (1 - \erf(\varepsilon/|a|)) > 0$ for the previous concentration inequality and suppose $|f(\bar v, \bm z) - \mathbb{E} f(\bar v, \bm z)| \leq \Delta$. Then for the optimal $\bv$,
\[
|\bar v| = |f(\bar v, \bm z)| \leq |f(\bar v, \bm z) - \mathbb{E} f(\bar v, \bm z)| + |\mathbb{E} f(\bar v, \bm z)| \leq \Delta + |\bar v| \cdot \erf(\varepsilon/|a|) \, ,
\]
where the first equality is from our original self-consistent equation $\bar v = f(\bar v, \bz)$ for the optimal $\bv$. This in turn implies $|\bar v| \leq \eta$. Finally, from our concentration inequality,
\[
\mathbb{P}(|\bar v^*(\bm z)| > \eta) \leq \exp(-C \eta^2 (1 - \erf(\varepsilon/|a|))^2 P/a^2) \, .
\]

\section{Regression Capacity of Random Sphere-Like Manifolds}\label{sec:si sphere-like manifolds}
To investigate how the geometric properties of manifolds connect to their regression capacity, we extend the random point-like manifold model in~\autoref{sec:si point-like manifolds} to random sphere-like manifolds with various latent geometric parameters including manifold dimension $D$, manifold radius $R$, correlation $\psi$ between manifold centers, correlation $\rho$ between manifold label, and correlation $\gamma$ between manifold internal axes. We derive closed-form formula (\autoref{thm:uncorelated spheres} and~\autoref{thm:corelated spheres}) for the regression capacity in terms of these latent parameters.

\subsection{Capacity of Uncorrelated Spheres}\label{sec:SI uncor spheres}
We introduce structured variability to our data. In particular, we consider $D$-dimensional spherical manifolds with radius $R$, where $D$ is fixed and $N, P \to \infty$. For simplicity, we will consider the uncorrelated case first. In the following section, we will consider a model of correlated spheres and reduce it to the uncorrelated spheres model, much like we reduce the correlated points model to uncorrelated points.

\subsubsection{Our toy model for uncorrelated spheres and its regression capacity}
To sample $P$ manifolds $\{\mathcal{M}^\mu \subseteq \mathbb{R}^N\}_{\mu=1}^P$ embedded in $\mathbb{R}^N$, we first sample their centroids $\{\mathbf{x}^\mu_0 \in \mathbb{R}^N\}_{\mu=1}^P$ with mean square-norm $\mathbb{E} \|\mathbf{x}^\mu_0\|^2 = r^2$ set by some $r > 0$. For each manifold $\mu \in \{1, 2, \dots, P\}$, we also sample $D$ i.i.d. axis vectors $\{\mathbf{x}^\mu_{i} \in \mathbb{R}^N\}_{i=1}^D$ with mean square-norm $\mathbb{E} \|\mathbf{x}_i^\mu\|^2 = R^2$ set by some $R > 0$. More concretely, the centroids are sampled from $x^{\mu}_{0k} \overset{i.i.d.}\sim \mathcal{N}\left(0, \frac{r^2}{N}\right)$ for all $\mu \in \{1, \dots, P\}$ and $k \in \{1, \dots, N\}$ and the axes are sampled from $x^{\mu}_{ik} \overset{i.i.d.}\sim \mathcal{N}\left(0, R^2/N\right)$ for all  $\mu \in \{1, \dots, P\}, i \in \{1, \dots, D\}$ and $k \in \{1, \dots, N\}$.
This scaling ensures that $\mathbb{E} \mathbf{x}_0^{\mu} \cdot \mathbf{x}_0
^\nu = r^2\delta_{\mu \nu}$, and $\mathbb{E} \mathbf{x}_i^{\mu} \cdot \mathbf{x}_j
^\nu = R^2\delta_{\mu \nu}\delta_{ij}$.
We sample the target labels $\{y^\mu\}_{\mu=1}^P$ for each point i.i.d. from $\mathcal{N}(0, \sigma^2)$ with some label variance $\sigma^2 > 0$. In the notation of~\autoref{sec:si mean-field}, we have manifold shape
$$\mathcal{S} = \left\{(1, s_1, \dots, s_D) : \sum_{i=1}^D s_i^2 = 1\right\} \subseteq 
\mathbb{R}^{D+1}$$
with dimension $D$, a correlation tensor $\Sigma \in \mathbb{R}^{P(D+1) \times P(D+1)}$ with entries
$$\Sigma_{\mu i}^{\nu j} = \begin{cases}
r^2\delta_{\mu \nu} & \textrm{if i = j = 0}, \\
R^2\delta_{\mu \nu} & \textrm{if i = j > 0},
\end{cases}$$
and label correlation $H \in \mathbb{R}^{P \times P}$ with entries $H_{\mu \nu} = \sigma^2\delta_{\mu \nu}$.

Here, we derive a simple capacity formula for the toy model described above by reducing the capacity formula for a correlated point model to an equivalent uncorrelated point model.
\begin{theorem}[Formula for uncorrelated spheres]\label{thm:uncorelated spheres}
Given the model for uncorrelated spheres described above. Let $\alpha^{-1}(\varepsilon, r,R,D,\sigma,P)$ be the capacity formula from~\autoref{corrcap} for this model, we have that $\alpha^{-1}(\varepsilon,r,R,D,\sigma,P)\xrightarrow{P \to \infty}\alpha^{-1}(\varepsilon,r,R,D,\sigma)$ with
\begin{equation}\label{eq:SI uncorr-spheres-cap}
\alpha^{-1}(\varepsilon,r,R,D,\sigma) = 2 \min_{\tau \geq 0} \int_{0}^\infty dx \; \chi_D(x) g(x, \tau, \varepsilon,r,R,\sigma)
\end{equation}
where $\chi_D(x) = \frac{x^{D-1}e^{-x^2/2}}{2^{D/2-1}\Gamma(D/2)} \mathbb{I}(x \geq 0)$ denotes the PDF of the $\chi$ distribution with $D$ degrees of freedom, and
\begin{align*}
    g(x, \tau, \varepsilon,r,R,\sigma) &= h\left(\sqrt{1+\frac{\sigma^2\tau^2}{r^2}}, \frac{\varepsilon\tau}{r}, \frac{r^2x/R+\varepsilon\tau}{\sqrt{r^2+\sigma^2\tau^2}}\right) \\
    &+ \frac{1}{2} x^2 \operatorname{erfc}\left(\frac{r^2x/R+\varepsilon\tau}{\sqrt{2(r^2+\sigma^2\tau^2)}}\right) \\
    &+ h\left(\frac{\sqrt{r^2+\sigma^2\tau^2}}{\sqrt{R^2+r^2}}, \frac{\varepsilon\tau - R x}{\sqrt{R^2+r^2}}, \max\left\{\frac{\varepsilon\tau - R x}{\sqrt{r^2+\sigma^2\tau^2}}, \frac{r^2 x/R - r^2 \varepsilon \tau / R^2}{\sqrt{r^2+\sigma^2\tau^2}}\right\}\right) \\
    &- h\left(\frac{\sqrt{r^2+\sigma^2\tau^2}}{\sqrt{R^2+r^2}}, \frac{\varepsilon\tau - R x}{\sqrt{R^2+r^2}}, \frac{r^2x/R+\varepsilon\tau}{\sqrt{r^2+\sigma^2\tau^2}}\right) \\
    &+ h\left(\sqrt{1+\frac{\sigma^2\tau^2}{r^2}}, 0, 0\right) \\
    &- h\left(\sqrt{1+\frac{\sigma^2\tau^2}{r^2}}, 0, \max\left\{ 0, \frac{r^2x/R-r^2\varepsilon\tau/R^2}{\sqrt{r^2+\sigma^2\tau^2}} \right\}\right) \\
    &+ \frac{1}{2} \left(x - \frac{\varepsilon\tau}{R}\right)^2 \operatorname{erf}\left(\max\left\{0, \frac{r^2x/R-r^2\varepsilon\tau/R^2}{\sqrt{2(r^2+\sigma^2\tau^2)}}\right\}\right)
\end{align*}
where
\begin{align*}
    h(a, b, c) &= \int_{c}^\infty \mathcal{D} z (az - b)^2 \\
    &= \frac{a(ac - 2b)}{\sqrt{2\pi}} \exp\left(-\frac{c^2}{2}\right) + \frac{a^2+b^2}{2} \operatorname{erfc}\left(\frac{c}{\sqrt{2}}\right) \, .
\end{align*}
\end{theorem}
We interpret this formula and numerically validate it in the main text. The derivation details for~\autoref{thm:uncorelated spheres} is provided in the following subsubsection for the completeness of presentation.

\subsubsection{Derivation details for~\autoref{thm:uncorelated spheres}}
For a given error tolerance $\varepsilon > 0$ and for each $P$,~\autoref{corrcap} gives us a capacity estimator
\[
\alpha^{-1}(\varepsilon, P) = \frac{1}{P} \min_{\tau \geq 0} \mathbb{E}_{\bm{t},\bm{y}} \min_{\bm{u} \in \mathcal{B}(\bm{y})} \left\Vert \bm{t} - \tau\Sigma^{-\frac{1}{2}}\bm{u}\right\Vert^2 \, .
\]
Our goal in this section will be to show that the limit
\[
\alpha^{-1}(\varepsilon) = \lim_{P \to \infty} \alpha^{-1}(\varepsilon, P)
\]
exists and compute it. Then we use~\autoref{corrcap} to conclude that $\alpha(\varepsilon)$ is the asymptotic capacity of our uncorrelated points model. For our specific model, we can simplify the above capacity estimator to:
$$\alpha^{-1}(\varepsilon, P) = \min_{\tau \geq 0} \frac{1}{P} \mathbb{E}_{\bm{t},\bm{y}} \min_{\bm{u} \in \mathcal{B}(\bm{y})} \left\Vert \bm{t}_0 - \frac{\tau}{r} \bm{u}_0 \right\Vert^2 + \sum_{i=1}^D \left\Vert \bm{t}_i - \frac{\tau}{R} \bm{u}_i \right\Vert^2,$$
where
$$\mathcal{B}(\bm{y}) = \left\{\bm{u} \in \mathbb{R}^{P(D+1)} : |u^{\mu}_{0} - y_\mu| \leq \varepsilon - \sqrt{\sum_{i=1}^D (u^{\mu}_{i})^2} \textrm{ for all } \mu \in \{1, \dots, P\}\right\}.$$

As in the uncorrelated points case, since the distributions of $\bm{t}$ and $\bm{y}$ as well as the constraint set $\mathcal{B}(\bm{y})$ all factorize across the manifolds, we can immediately reduce our capacity estimator to the $P=1$ case:
$$\alpha^{-1}(\varepsilon, P) = \alpha^{-1}(\varepsilon, 1) = \min_{\tau \geq 0} \mathbb{E}_{\bm{t},y} \min_{\bm{u} \in \mathcal{B}(y)} \left(t_0 - \frac{\tau}{r} u_0 \right)^2 + \sum_{i=1}^D \left\Vert t_i - \frac{\tau}{R} u_i\right\Vert^2,$$
where now the components $\{t_i \in \mathbb{R}\}_{i=0}^D$ and $\{u_i \in \mathbb{R}\}_{i=0}^D$ are scalars and the constraint set is just
$$\mathcal{B}(y) = \left\{\bm{u} \in \mathbb{R}^P : |u_{0} - y| \leq \varepsilon - \sqrt{\sum_{i=1}^D u_{i}^2}\right\}.$$
As in previous sections, we make a change of variables
$$v_i = \begin{cases}
    u_0 - y & \textrm{if } i = 0, \\
    u_i & \textrm{if } i > 0,
\end{cases}$$
so that capacity becomes
$$\alpha^{-1}(\varepsilon) = \min_{\tau \geq 0} \mathbb{E}_{\bm{t},y} \min_{\bm{v} \in \mathcal{B}(0)} \left(t_0 - \frac{\tau}{r}(v_0 + y)\right)^2 + \sum_{i=1}^D \left\Vert t_i - \frac{\tau}{R} v_i \right\Vert^2.$$
As before, we can now combine $\bm{t}$ and $\bm{y}$ so that our estimator becomes
$$\alpha^{-1}(\varepsilon) = \min_{\tau \geq 0} \mathbb{E}_{\bm{g}} f_\tau(\bm{g}),$$
where we define the field $f_\tau(\bm{g})$ by
$$f_\tau(\bm{g}) = \min_{\bm{v} \in \mathcal{B}(0)} \left(\tau g_0 - \frac{\tau}{r}v_0\right)^2 + \sum_{i=1}^D \left\Vert \tau g_i - \frac{\tau}{R} v_i \right\Vert^2,$$
and where $\bm{g} \in \mathbb{R}^{D+1}$ is a mean-zero Gaussian vector with covariance
$$\mathbb{E} g_i g_j = \begin{cases}
    \frac{1}{\tau^2} + \frac{\sigma^2}{r^2}& \textrm{if } i = j = 0, \\
    \frac{1}{\tau^2} & \textrm{if } i = j > 0.
\end{cases}$$
Let us make the further change of variables
$$w_i = \begin{cases}
    \frac{v_0}{r} & \textrm{if } i = 0, \\
    \frac{v_i}{R} & \textrm{if } i > 0,
\end{cases}$$
so that
$$f_\tau(\bm{g}) = \tau^2 \min_{\bm{w} \in \mathcal{A}} \|\bm{g} - \bm{w}\|^2,$$
where
$$\mathcal{A} = \left\{\bm{w} \in \mathbb{R}^{D+1} : r|w_0| \leq \varepsilon - R \sqrt{\sum_{i=1}^D w_i^2}\right\}.$$
The minimization problem in $f_\tau(\bm{g})$ is just a projection of $\bm{g}$ onto the set $\mathcal{A}$. For $D = 1$, we can visualize this solution set $\mathcal{A}$ to get some intuition for how projection behaves:

\begin{figure}[H]
    \centering
    \includegraphics[width=0.5\linewidth]{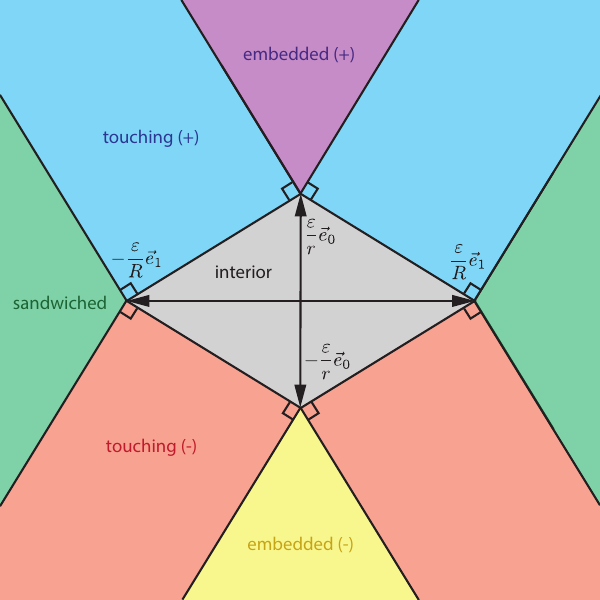}
    \caption{Visualization of the solution set $\mathcal{A}$.}
\end{figure}

\noindent
The diamond-shaped gray convex body in the center is the solution set $\mathcal{A}$. We have split the space $\mathbb{R}^{D+1}$ into six different regimes based on how the projection behaves when $\bm{g}$ is in each region of space. For example, if $\bm{g}$ is in the cone-shaped \enquote{embedded (+)} regime, then the nearest solution $\bm{w} \in \mathcal{A}$ is always the point $\frac{\varepsilon}{r} \vec e_0$ (where $\{\vec e_i\}$ is the standard basis for $\mathbb{R}^{D+1}$). To compute $\mathbb{E}_{\bm{g}} f_\tau(\bm{g})$, we split into these six cases, which we now formally define: Going back to the general case of arbitrary $D$, we can write down a Lagrangian characterizing the optimal $W$ in the definition of $f_\tau(\bm{g})$:
$$\mathcal{L}(\bm{w}, \lambda^+, \lambda^-) = \frac{1}{2} \|\bm{g} - \bm{w}\|^2 + \lambda^+ (rw_0 - \varepsilon + R \|\bm{\tilde w}\|) + \lambda^- (-rw_0 - \varepsilon + R \|\bm{\tilde w}\|),$$
setting $\bm{\tilde w} = (w_1, w_2, \dots, w_D)$ and $\bm{\tilde g} = (g_1, g_2, \dots, g_D)$ for convenience. Recall that the subgradient of the function $\bm{x} \mapsto \|\bm{x}\|$ at the point $\bm{x} = 0$ consists of all vectors $\bm{v}$ with $\|\bm{v}\| \leq 1$. The KKT conditions for this Lagrangian are then:

\begin{itemize}
    \item $\frac{\partial \mathcal{L}}{\partial w_0} = w_0 - g_0 + (\lambda^+ - \lambda^-) r = 0$ (stationarity w.r.t $w_0$), \\
    \item Either $\bm{\tilde w} \neq 0$ and $\nabla_{\bm{\tilde w}}\mathcal{L} = \bm{\tilde w} - \bm{\tilde g} + (\lambda^+ + \lambda^-) R \frac{\bm{\tilde w}}{\|\bm{\tilde w}\|} = 0$, or $\bm{\tilde w} = 0$ and $\|\bm{\tilde g}\| \leq (\lambda^+ + \lambda^-) R$ (stationarity w.r.t. $\bm{\tilde W}$), \\
    \item If $\lambda^+ > 0$, then $rw_0 = \varepsilon - R \|\bm{\tilde w}\|$ (complementary slackness of upper bound), \\
    \item If $\lambda^- > 0$, then $-rw_0 = \varepsilon - R \|\bm{\tilde w}\|$ (complementary slackness of lower bound),
\end{itemize}
plus primal feasibility $\bm{w} \in \mathcal{A}$ and dual feasibility $\lambda^+, \lambda^- \geq 0$. We now split into six cases based on which of the variables $\lambda^+, \lambda^-, \bm{\tilde w}$ are zero, characterize the values of $\bm{g}$ corresponding to each case (Equations C1-C6), and obtain closed-form expressions for $f_\tau(\bm{g})$ in each case (Equations F1-F6):

\vspace{1em}
\noindent \emph{Case 1: \enquote{Interior regime} when $\lambda^+ = \lambda^- = 0$}
\vspace{1em}

\noindent
In this regime, the stationarity conditions immediately yields $\bm{w} = \bm{g}$, so 
\begin{equation}\tag{F1}
    f_\tau(\bm{g}) = 0
\end{equation}
Since $\bm{g} = \bm{w} \in \mathcal{A}$, we have
\begin{equation}\tag{C1}
    |g_0| \leq \frac{\varepsilon}{r} - \frac{R\|\bm{\tilde g}\|}{r}.
\end{equation}

\vspace{1em}
\noindent \emph{Case 2: \enquote{Embedded (+) regime} when $\lambda^+ > 0, \lambda^- = 0, \bm{\tilde w} = 0$}
\vspace{1em}

\noindent
In this regime, complementary slackness tells us $rw_0 = \varepsilon - R \|\bm{\tilde w}\|$, and since $\bm{\tilde w} = 0$, $w_0 = \varepsilon/r$. The stationarity conditions yield $g_0 = w_0 + \lambda^+ r$ and $\|\bm{\tilde g}\| \leq \lambda^+ R$. Combining these, we arrive at
\begin{equation}\tag{C2}
    g_0 \geq \frac{\varepsilon}{r} + \frac{r\|\bm{\tilde g}\|}{R}.
\end{equation}
Also, we can write
\begin{equation}\tag{F2}
f_\tau(\bm{g}) = \tau^2\left(g_0 - \frac{\varepsilon}{r}\right)^2 + \tau^2\|\bm{\tilde g}\|^2.
\end{equation}

\vspace{1em}
\noindent \emph{Case 3: \enquote{Embedded (-) regime} when $\lambda^+ = 0, \lambda^- > 0, \bm{\tilde w} = 0$}
\vspace{1em}

\noindent
In this regime, complementary slackness tells us $-rw_0 = \varepsilon - R \|\bm{\tilde w}\|$, and since $\bm{\tilde w} = 0$, $w_0 = -\varepsilon/r$. The stationarity conditions yield $g_0 = w_0 - \lambda^- r$ and $\|\bm{\tilde g}\| \leq \lambda^- R$. Combining these, we arrive at
\begin{equation}\tag{C3}
    g_0 \leq -\frac{\varepsilon}{r} - \frac{r\|\bm{\tilde g}\|}{R}.
\end{equation}
Also, we can write
\begin{equation}\tag{F3}
f_\tau(\bm{g}) = \tau^2\left(g_0 + \frac{\varepsilon}{r}\right)^2 + \tau^2\|\bm{\tilde g}\|^2.
\end{equation}

\vspace{1em}
\noindent \emph{Case 4: \enquote{Touching (+) regime} when $\lambda^+ > 0, \lambda^- = 0, \bm{\tilde w} \neq 0$}
\vspace{1em}

\noindent
In this regime, stationarity gives $g_0 = w_0 + \lambda^+ r \geq w_0$. Complementary slackness gives $w_0 = \frac{\varepsilon}{r} - \frac{R \|\bm{\tilde w}\|}{r}$. From stationarity, $\|\bm{\tilde w}\| = \|\bm{\tilde g}\| - \lambda^+ R \leq \|\bm{\tilde g}\|$. Combining these, we get
$$g_0 \geq \frac{\varepsilon}{r} - \frac{R \|\bm{\tilde g}\|}{r}.$$

Stationarity also gives $R \|\bm{\tilde g}\| = R \|\bm{\tilde w}\| + \lambda^+ R^2$. From complementary slackness, we also know $R \|\bm{\tilde w}\| = \varepsilon - r w_0$. Combining these, $R \|\bm{\tilde g}\| = \varepsilon - r w_0 + \lambda^+ R^2$. From stationarity, we know $rw_0 = r g_0 - \lambda^+ r^2$. Together this yields $R \|\bm{\tilde g}\| = \varepsilon - r g_0 + \lambda^+ [r^2 + R^2]$. We can then solve for $\lambda^+$:
$$\lambda^+ = \frac{R \|\bm{\tilde g}\| - \varepsilon + r g_0}{R^2 + r^2}.$$
From stationarity, we know $\lambda^+ = \frac{\|\bm{\tilde g}\|}{R} - \frac{\|\bm{\tilde w}\|}{R} \leq \frac{\|\bm{\tilde g}\|}{R}$. Thus,
$$\lambda^+ = \frac{R \|\bm{\tilde g}\| - \varepsilon + r g_0}{R^2 + r^2} \leq \frac{\|\bm{\tilde g}\|}{R}.$$
This inequality rearranges to give
$$g_0 \leq \frac{\varepsilon}{r} + \frac{r \|\bm{\tilde g}\|}{R}.$$

From stationarity, we have $\|\bm{\tilde g}\| = \|\bm{\tilde w}\| + \lambda^+ R$. From complementary slackness, $\|\bm{\tilde w}\| = \frac{\varepsilon}{R} - \frac{r w_0}{R}$. Combining these, $\|\bm{\tilde g}\| = \frac{\varepsilon}{R} - \frac{r w_0}{R} + \lambda^+ R$. Since we have $\lambda^+ = \frac{g_0 - w_0}{r}$ from stationarity, it follows that
$$\|\bm{\tilde g}\| = \frac{\varepsilon}{R} - \frac{r w_0}{R} + \frac{Rg_0}{r} - \frac{Rw_0}{r}.$$
Since $|r w_0| \leq \varepsilon - R \|\bm{\tilde w}\|$ from primal feasibility, certainly $\varepsilon - R \|\bm{\tilde w}\| \geq 0$, so from slackness we have $w_0 \geq 0$. Combining this with the above inequality,
$$\|\bm{\tilde g}\| \leq \frac{\varepsilon}{R} + \frac{R g_0}{r},$$
which rearranges to
$$g_0 \geq \frac{r\|\bm{\tilde g}\|}{R} - \frac{r\varepsilon}{R^2}.$$
Combining all these results, we arrive at
\begin{equation}\tag{C4}
    \max\left\{\frac{\varepsilon}{r} - \frac{R \|\bm{\tilde g}\|}{r}, \frac{r \|\bm{\tilde g}\|}{R} - \frac{r\varepsilon}{R^2}\right\} \leq g_0 \leq \frac{\varepsilon}{r} + \frac{r\|\bm{\tilde g}\|}{R}.
\end{equation}

From stationarity, $\bm{\tilde g} - \bm{\tilde w} = \lambda^+ R \frac{\bm{\tilde w}}{\|\bm{\tilde w}\|}$, so $\|\bm{\tilde g} - \bm{\tilde w}\|^2 = (\lambda^+)^2 R^2$. Also from stationarity, $g_0 - w_0 = \lambda^+ r$, so $(g_0 - w_0)^2 = (\lambda^+)^2 r^2$. Combining this with our equation for $\lambda^+$ above, we get
\begin{equation}\tag{F4}
f_\tau(\bm{g}) = \tau^2 \frac{(R\|\bm{\tilde g}\| - \varepsilon + r g_0)^2}{R^2+r^2}.
\end{equation}

\vspace{1em}
\noindent \emph{Case 5: \enquote{Touching (-) regime} when $\lambda^+ = 0, \lambda^- > 0, \bm{\tilde w} \neq 0$}
\vspace{1em}

\noindent
In this regime, stationarity gives $g_0 = w_0 - \lambda^- r \leq w_0$. Complementary slackness gives $w_0 = -\frac{\varepsilon}{r} + \frac{R \|\bm{\tilde w}\|}{r}$. From stationarity, $\|\bm{\tilde w}\| = \|\bm{\tilde g}\| - \lambda^- R \leq \|\bm{\tilde g}\|$. Combining these, we get
$$g_0 \leq -\frac{\varepsilon}{r} + \frac{R \|\bm{\tilde g}\|}{r}.$$

Stationarity also gives $R \|\bm{\tilde g}\| = R \|\bm{\tilde w}\| + \lambda^- R^2$. From complementary slackness, we also know $R \|\bm{\tilde w}\| = \varepsilon + r w_0$. Combining these, $R \|\bm{\tilde g}\| = \varepsilon + r w_0 + \lambda^+ R^2$. From stationarity, we know $rw_0 = r g_0 + \lambda^- r^2$. Together this yields $R \|\bm{\tilde g}\| = \varepsilon + r g_0 + \lambda^- [r^2 + R^2]$. We can then solve for $\lambda^-$:
$$\lambda^- = \frac{R \|\bm{\tilde g}\| - \varepsilon - r g_0}{R^2 + r^2}.$$
From stationarity, we know $\lambda^- = \frac{\|\bm{\tilde g}\|}{R} - \frac{\|\bm{\tilde w}\|}{R} \leq \frac{\|\bm{\tilde g}\|}{R}$. Thus,
$$\lambda^- = \frac{R \|\bm{\tilde g}\| - \varepsilon - r g_0}{R^2 + r^2} \leq \frac{\|\bm{\tilde g}\|}{R}.$$
This inequality rearranges to give
$$g_0 \geq -\frac{\varepsilon}{r} - \frac{r \|\bm{\tilde g}\|}{R}.$$

From stationarity, we have $\|\bm{\tilde g}\| = \|\bm{\tilde w}\| + \lambda^- R$. From complementary slackness, $\|\bm{\tilde w}\| = \frac{\varepsilon}{R} + \frac{r w_0}{R}$. Combining these, $\|\bm{\tilde g}\| = \frac{\varepsilon}{R} + \frac{r w_0}{R} + \lambda^- R$. Since we have $\lambda^- = \frac{-g_0 + w_0}{r}$ from stationarity, it follows that
$$\|\bm{\tilde g}\| = \frac{\varepsilon}{R} + \frac{r w_0}{R} - \frac{Rg_0}{r} + \frac{Rw_0}{r}.$$
Since $|r w_0| \leq \varepsilon - R \|\bm{\tilde w}\|$ from primal feasibility, certainly $\varepsilon - R \|\bm{\tilde w}\| \geq 0$, so from slackness we have $w_0 \leq 0$. Combining this with the above inequality,
$$\|\bm{\tilde g}\| \leq \frac{\varepsilon}{R} - \frac{R g_0}{r},$$
which rearranges to
$$g_0 \leq -\frac{r\|\bm{\tilde g}\|}{R} + \frac{r\varepsilon}{R^2}.$$
Combining all these results, we arrive at
\begin{equation}\tag{C5}
    -\frac{\varepsilon}{r} - \frac{r\|\bm{\tilde g}\|}{R} \leq g_0 \leq \min\left\{-\frac{\varepsilon}{r} + \frac{R \|\bm{\tilde g}\|}{r}, -\frac{r \|\bm{\tilde g}\|}{R} + \frac{r\varepsilon}{R^2}\right\}.
\end{equation}

From stationarity, $\bm{\tilde g} - \bm{\tilde w} = \lambda^- R \frac{\bm{\tilde w}}{\|\bm{\tilde w}\|}$, so $\|\bm{\tilde g} - \bm{\tilde w}\|^2 = (\lambda^-)^2 R^2$. Also from stationarity, $g_0 - w_0 = -\lambda^- r$, so $(g_0 - w_0)^2 = (\lambda^-)^2 r^2$. Combining this with our equation for $\lambda^-$ above, we get
\begin{equation}\tag{F5}
f_\tau(\bm{g}) = \tau^2 \frac{(R\|\bm{\tilde g}\| - \varepsilon - r g_0)^2}{R^2+r^2}.
\end{equation}

\vspace{1em}
\noindent \emph{Case 6: \enquote{Sandwiched regime} when $\lambda^+ > 0, \lambda^- > 0$}
\vspace{1em}

\noindent
In this regime, complementary slackness implies \emph{both} $rw_0 = \varepsilon - R \|\bm{\tilde w}\|$ and $-rw_0 = \varepsilon - R \|\bm{\tilde w}\|$. Thus, $w_0 = 0$ and $\|\bm{\tilde w}\| = \varepsilon/R$. Since $\bm{\tilde w} \neq 0$, stationarity gives $\|\bm{\tilde g}\| = \|\bm{\tilde w}\| + (\lambda^+ + \lambda^-) R$. Since $\|\bm{\tilde w}\| = \varepsilon/r$, this implies $(\lambda^+ + \lambda^-)R = \|\bm{\tilde g}\| - \varepsilon/R$. From stationarity again, we have $g_0 = w_0 + (\lambda^+ - \lambda^-) r$, and since $w_0 = 0$, we have $|g_0| = |\lambda^+ - \lambda^-| r \leq (\lambda^+ + \lambda^-) r$. Combining these,

\begin{equation}\tag{C6}
    |g_0|  \leq \frac{r \|\bm{\tilde g}\|}{R} - \frac{r\varepsilon}{R^2}.
\end{equation}

Since $w_0 = 0$, $(g_0 - w_0)^2 = g_0^2$. From stationarity, we have $\|\bm{\tilde g} - \bm{\tilde w}\|^2 = (\lambda^+ + \lambda^-)^2 R^2$. We showed above that $(\lambda^+ + \lambda^-)R = \|\bm{\tilde g}\| - \varepsilon/R$, so
\begin{equation}\tag{F6}
    f_\tau(\bm{g}) = \tau^2 g_0^2 + \tau^2 \left(\|\bm{\tilde g}\| - \frac{\varepsilon}{R}\right)^2.
\end{equation}

\noindent
Finally, observe that equations C1-C6 are mutually exclusive (ignoring overlap of probability zero). Thus, conditions C1-C6 completely characterize the six regimes we identified. Using the explicit formulae for $f_\tau(\bm{g})$ in terms of $g_0$ and $\|\bm{\tilde g}\|$ we derived above for each regime, we can write the average $\mathbb{E}_{\bm{g}} f_\tau(\bm{g})$ as:
\begin{align*}
    \mathbb{E}_G f_\tau(\bm{g}) &= \mathbb{E} \left[ 0 \middle\vert |g_0| \leq \frac{\varepsilon}{r} - \frac{R \|\bm{\tilde G}\|}{r} \right] \\
    &+ \mathbb{E} \left[ \tau^2\left(g_0 - \frac{\varepsilon}{r}\right)^2 + \tau^2\|\bm{\tilde g}\|^2 \middle\vert g_0 \geq \frac{\varepsilon}{r} + \frac{r \|\bm{\tilde g}\|}{R} \right] \\
    &+ \mathbb{E} \left[ \tau^2\left(g_0 + \frac{\varepsilon}{r}\right)^2 + \tau^2\|\bm{\tilde G}\|^2 \middle\vert g_0 \leq -\frac{\varepsilon}{r} - \frac{r \|\bm{\tilde g}\|}{R} \right] \\
    &+ \mathbb{E} \left[ \tau^2 \frac{(R\|\bm{\tilde g}\| - \varepsilon + r g_0)^2}{R^2+r^2} \middle\vert \max\left\{\frac{\varepsilon}{r} - \frac{R \|\bm{\tilde g}\|}{r}, \frac{r \|\bm{\tilde g}\|}{R} - \frac{r\varepsilon}{R^2}\right\} \leq g_0 \leq \frac{\varepsilon}{r} + \frac{r\|\bm{\tilde g}\|}{R} \right] \\
    &+ \mathbb{E} \left[ \tau^2 \frac{(R\|\bm{\tilde g}\| - \varepsilon - r g_0)^2}{R^2+r^2} \middle\vert -\frac{\varepsilon}{r} - \frac{r\|\bm{\tilde g}\|}{R} \leq g_0 \leq \min\left\{-\frac{\varepsilon}{r} + \frac{R \|\bm{\tilde g}\|}{r}, -\frac{r \|\bm{\tilde g}\|}{R} + \frac{r\varepsilon}{R^2}\right\} \right] \\
    &+ \mathbb{E} \left[ \tau^2 g_0^2 + \tau^2 \left(\|\bm{\tilde g}\| - \frac{\varepsilon}{R}\right)^2 \middle\vert |g_0|  \leq \frac{r \|\bm{\tilde g}\|}{R} - \frac{r\varepsilon}{R^2} \right]. \\
\end{align*}
Noting that $\tau\|\bm{\tilde g}\| \sim \chi(D)$ and $g_0/\sqrt{\frac{1}{\tau^2}+\frac{\sigma^2}{r^2}} \sim \mathcal{N}(0, 1)$, we can rewrite these expectations as integrals:
\begin{align*}
    \mathbb{E}_{\bm{g}} f_\tau(\bm{g}) &= 2\int_0^\infty dx \chi_D(x) \int_{\frac{r^2 x / R + \varepsilon \tau}{\sqrt{r^2+\sigma^2\tau^2}}}^\infty \mathcal{D}z \; \left(\sqrt{1+\frac{\sigma^2\tau^2}{r^2}} z - \frac{\varepsilon\tau}{r}\right)^2 + x^2 \\
    &+ 2 \int_0^\infty dx \chi_D(x) \int_{\max\left\{\frac{\varepsilon\tau - R x}{\sqrt{r^2+\sigma^2\tau^2}}, \frac{r^2 x/R - r^2 \varepsilon \tau / R^2}{\sqrt{r^2+\sigma^2\tau^2}}\right\}}^{\frac{r^2 x/R + \varepsilon\tau}{\sqrt{r^2+\sigma^2\tau^2}}} \mathcal{D}z\;  \frac{(R x - \varepsilon \tau + \sqrt{r^2+\sigma^2\tau^2} z)^2}{R^2+r^2} \\
    &+ 2\int_0^\infty dx \chi_D(x) \int_{0}^{\max\left\{0,\frac{r^2 x / R - r^2 \varepsilon \tau/R^2}{\sqrt{r^2+\sigma^2\tau^2}}\right\}} \mathcal{D}z \; \sqrt{1+\frac{\sigma^2\tau^2}{r^2}} z^2 + \left(x - \frac{\varepsilon\tau}{R}\right)^2
\end{align*}
where
$$\chi_D(x) = \frac{x^{D-1}e^{-x^2/2}}{2^{D/2-1}\Gamma(D/2)} \mathbb{I}(x \geq 0)$$
is the PDF of the $\chi$ distribution with $D$ degrees of freedom and $\mathcal{D}z$ denotes integration with respect to the standard Gaussian measure on $z$. We have also made use of some symmetry in the problem to combine the embedded (+) and embedded (-), as well as the touching (+) and touching (-) regimes. Finally, we can compute the capacity of our uncorrelated spheres model via:
$$\alpha^{-1}(\varepsilon) = \min_{\tau \geq 0} \mathbb{E}_{\bm{g}} f_\tau(\bm{g}).$$
The inner integrals have closed forms in terms of standard functions, but the outer integrals must be evaluated numerically. In particular, we can write
$$\alpha^{-1}(\varepsilon) = 2 \min_{\tau \geq 0} \int_{0}^\infty dx \; \chi_D(x) g(x, \tau, \varepsilon),$$
where
\begin{equation}\label{eq:uncorr-sphere-cap}
\begin{aligned}
    g(x, \tau, \varepsilon) &= h\left(\sqrt{1+\frac{\sigma^2\tau^2}{r^2}}, \frac{\varepsilon\tau}{r}, \frac{r^2x/R+\varepsilon\tau}{\sqrt{r^2+\sigma^2\tau^2}}\right) \\
    &+ \frac{1}{2} x^2 \operatorname{erfc}\left(\frac{r^2x/R+\varepsilon\tau}{\sqrt{2(r^2+\sigma^2\tau^2)}}\right) \\
    &+ h\left(\frac{\sqrt{r^2+\sigma^2\tau^2}}{\sqrt{R^2+r^2}}, \frac{\varepsilon\tau - R x}{\sqrt{R^2+r^2}}, \max\left\{\frac{\varepsilon\tau - R x}{\sqrt{r^2+\sigma^2\tau^2}}, \frac{r^2 x/R - r^2 \varepsilon \tau / R^2}{\sqrt{r^2+\sigma^2\tau^2}}\right\}\right) \\
    &- h\left(\frac{\sqrt{r^2+\sigma^2\tau^2}}{\sqrt{R^2+r^2}}, \frac{\varepsilon\tau - R x}{\sqrt{R^2+r^2}}, \frac{r^2x/R+\varepsilon\tau}{\sqrt{r^2+\sigma^2\tau^2}}\right) \\
    &+ h\left(\sqrt{1+\frac{\sigma^2\tau^2}{r^2}}, 0, 0\right) \\
    &- h\left(\sqrt{1+\frac{\sigma^2\tau^2}{r^2}}, 0, \max\left\{ 0, \frac{r^2x/R-r^2\varepsilon\tau/R^2}{\sqrt{r^2+\sigma^2\tau^2}} \right\}\right) \\
    &+ \frac{1}{2} \left(x - \frac{\varepsilon\tau}{R}\right)^2 \operatorname{erf}\left(\max\left\{0, \frac{r^2x/R-r^2\varepsilon\tau/R^2}{\sqrt{2(r^2+\sigma^2\tau^2)}}\right\}\right).
\end{aligned}
\end{equation}
and we define
\begin{align*}
    h(a, b, c) &= \int_{c}^\infty \mathcal{D} z (az - b)^2 \\
    &= \frac{a(ac - 2b)}{\sqrt{2\pi}} \exp\left(-\frac{c^2}{2}\right) + \frac{a^2+b^2}{2} \operatorname{erfc}\left(\frac{c}{\sqrt{2}}\right).
\end{align*}
We interpret this formula and numerically validate it in the main text.

\subsection{Correlated Spheres Reduce to Uncorrelated Spheres}\label{sec:SI corrsphere}
Next, we introduce uniform correlations to our spheres model. As in our correlated points model, we have parameters $\psi \in [0, 1)$ and $\rho \in [0, 1)$ that control the correlations among the sphere centers and labels, respectively. Additionally, we introduce another correlation $\gamma \in [0, 1)$ among the \emph{axes} along which the spheres are embedded into $\mathbb{R}^N$. 

\subsubsection{Our toy model for correlated spheres and its regression capacity}
In the notation of Theorem 5, we define the data correlation tensor by
$$\Sigma_{\mu i}^{\nu j} = \begin{cases}
    r^2[(1-\psi)\delta_{\mu\nu}+\psi] & \textrm{if } i = j = 0, \\
    R^2[(1-\gamma)\delta_{\mu\nu}+\gamma] & \textrm{if } i = j > 0
\end{cases}$$
and the label correlation matrix by
$$M_{\mu\nu} = \sigma^2[(1-\rho)\delta_{\mu\nu}+\rho].$$

Here, we derive a simple capacity formula for the toy model described above by reducing the capacity formula for a correlated point model to an equivalent uncorrelated point model.
\begin{theorem}[Formula for correlated spheres]\label{thm:corelated spheres}
Given the model for correlated spheres described above. Let $\alpha^{-1}(\varepsilon, r,R,D,\sigma,\rho,\psi,\gamma,P)$ be the capacity formula from~\autoref{corrcap} for this model, we have that $\alpha^{-1}(\varepsilon,r,R,D,\sigma,\rho,\psi,\gamma,P)\xrightarrow{P \to \infty}\alpha^{-1}(\varepsilon,r,R,D,\sigma,\rho,\psi,\gamma)$ with
\begin{equation}\label{eq:SI corr-spheres-cap}
\alpha^{-1}(\varepsilon,r,R,D,\sigma,\rho,\psi,\gamma) = \alpha^{-1}(\varepsilon,r\sqrt{1-\psi},R\sqrt{\gamma},D,\sigma\sqrt{1-\rho}) \, .
\end{equation}
\end{theorem}
We interpret this formula and numerically validate it in the main text. The derivation details for~\autoref{thm:corelated spheres} is provided in the following subsubsection for the completeness of presentation.

\subsubsection{Derivation details for~\autoref{thm:corelated spheres}}
For a given label tolerance $\varepsilon > 0$, \autoref{corrcap} (subject to the regularity conditions we prove below) gives us a capacity estimator
$$\alpha^{-1}(\varepsilon, P) = \min_{\tau \geq 0} \mathbb{E}_{\bm t, \bm y} \min_{\bm u \in \mathcal{B}(\bm y)} \frac{1}{P} \|\bm t - \tau\Sigma^{-\frac{1}{2}} \bm u\|^2,$$
where
$$\mathcal{B}(\bm y) = \left\{\bm u \in \mathbb{R}^{P(D+1)} : |u^\mu_0 - \bm y^\mu| \leq \varepsilon - \sqrt{\sum_{i=1}^D \bm (u^\mu_i)^2} \textrm{ for all } \mu \in \{1, \dots, P\}\right\},$$
where we have simplified $\mathcal{B}(\bm y)$ for the case of spherical manifold shape $\mathcal{S}$. Given the block structure of the correlation tensor $\Sigma$, we can write
$$\alpha^{-1}(\varepsilon, P) = \min_{\tau \geq 0} \mathbb{E}_{\bm t, \bm y} \min_{\bm u \in \mathcal{B}(\bm y)} \frac{1}{P} \left[\left\Vert \bm t_0 - \tau (\Sigma^{(0)})^{-\frac{1}{2}} \bm u_0 \right\Vert^2 + \sum_{i=1}^D \left\Vert \bm t_i - \tau (\Sigma^{(1)})^{-\frac{1}{2}} \bm u_i \right\Vert^2\right],$$
where we define $\bm t_i = (\bm t^{\mu}_i)_{\mu=1}^P \in \mathbb{R}^P$ for each $i$ (and similarly for $\bm u$) and $\Sigma^{(0)}_{\mu \nu} = r^2[(1-\psi)\delta_{\mu\nu}+\psi]$ and $\Sigma^{(1)}_{\mu\nu} = R^2[(1-\gamma)\delta_{\mu\nu}+\gamma]$ are the center and axis blocks of the data correlation tensor. As in previous sections, we can define a shifted variable $\bm v \in \mathbb{R}^{P(D+1)}$ by
$$\bm v^\mu_i = \begin{cases}
    \bm u^\mu_0 - \bm y_\mu & \textrm{if } i = 0, \\
    \bm u^\mu_i & \textrm{if } i > 0,
\end{cases}$$
and combine the randomness of $\bm t, \bm y$ into a single standard Gaussian vector $\bm z \in \mathbb{R}^{P(D+1)}$ so that our estimator becomes
$$\alpha^{-1}(\varepsilon, P) = \min_{\tau \geq 0} \mathbb{E}_{\bm t, \bm y} f_\tau(\bm z),$$
where
$$f_\tau(\bm z) = \frac{1}{P}\left[\left\Vert  (\Omega^{(0)})^{1/2} \bm z_0 - \tau (\Sigma^{(0)})^{-1/2} \bm v_0\right\Vert^2 + \sum_{i=1}^D \left\Vert \bm z_i - \tau (\Sigma^{(1)})^{-1/2} \bm v_i\right\Vert^2\right]$$
and $\Omega^{(0)} = I + \tau^2 (\Sigma^{(0)})^{-\frac{1}{2}}H(\Sigma^{(0)})^{-\frac{1}{2}}$.

As in our analysis of the correlated points model, we decompose the matrices $\Sigma^{(0)}, \Omega^{(0)}$, and $\Sigma^{(1)}$ into their eigencomponents, with one eigenspace spanned by $\mathbf{1}/\sqrt{P}$ and the other $(P-1)$-dimensional eigenspace orthogonal to it. Just as in the correlated points case, we will show that in some sense only the $(P-1)$-dimensional eigenspaces matter for computing capacity. Thus, we will be able to replace these correlation matrices with rescaled identity matrices and reduce the asymptotic capacity of this model to that of the uncorrelated spheres model with suitably rescaled parameters. Letting $\Pi = I - \frac{1}{P} \mathbbm{1}\mathbbm{1}^\top \in \mathbb{R}^{P \times P}$, the eigendecompositions of $(\Omega^{(0)})^{1/2}$, $(\Sigma^{(0)})^{-1/2}$, and $(\Sigma^{(1)})^{-1/2}$ are
\begin{align*}
    (\Omega^{(0)})^{1/2} &= \sqrt{1 + \tau^2 \frac{\sigma^2(1-\rho)}{r^2(1-\psi)}} \Pi + \sqrt{1+\tau^2 \frac{\sigma^2(1-\rho + P\rho)}{r^2(1-\psi + P\psi)}} \frac{1}{P} \mathbbm{1}\mathbbm{1}^\top \\
    (\Sigma^{(0)})^{-1/2} &= \frac{1}{r\sqrt{1-\psi}} \Pi + \frac{1}{r\sqrt{1-\psi+P\psi}} \frac{1}{P}\mathbbm{1}\mathbbm{1}^\top \\
    (\Sigma^{(1)})^{-1/2} &= \frac{1}{R\sqrt{1-\gamma}} \Pi + \frac{1}{R\sqrt{1-\gamma+P\gamma}} \frac{1}{P}\mathbbm{1}\mathbbm{1}^\top.
\end{align*}
we can decompose the field $f_\tau(\bm z)$ as
\begin{align*}
    f_\tau(\bm z) &= \min_{\bm v \in \mathcal{B}(0)}\bigg[ \frac{1}{P}\left\Vert \Pi (a_0 \bm z_0 - b_0 \bm v_0)\right\Vert^2 + \frac{1}{P} \left\Vert  \left[\frac{1}{P} \mathbbm{1}\mathbbm{1}^\top\right]( \tilde a_0 \bm z_0 - \tilde b_0 \bm v_0)\right\Vert^2 \\
    &+\sum_{i=1}^D  \frac{1}{P}\left\Vert \Pi (\bm z_i - b_1 \bm v_i)\right\Vert^2 + \sum_{i=1}^D \frac{1}{P} \left\Vert \left[\frac{1}{P} \mathbbm{1}\mathbbm{1}^\top\right]( \bm z_i - \tilde b_1 \bm v_i)\right\Vert^2 \bigg],
\end{align*}
where the coefficients above are defined as
\begin{align*}
    a_0 &= \sqrt{1 + \tau^2 \frac{\sigma^2(1-\rho)}{r^2(1-\psi)}}, \\
    b_0 &= \frac{\tau}{r\sqrt{1-\psi}} \\
    \tilde a_0 &= \sqrt{1+\tau^2 \frac{\sigma^2(1-\rho + P\rho)}{r^2(1-\psi + P\psi)}} \\
    \tilde b_0 &= \frac{\tau}{r\sqrt{1-\psi+P\psi}} \\
    b_1 &= \frac{\tau}{R\sqrt{1-\gamma}} \\
    \tilde b_1 &= \frac{\tau}{R\sqrt{1-\gamma+P\gamma}}.
\end{align*}
We can write
\begin{align*}
     \frac{1}{P} \left\Vert  \left[\frac{1}{P} \mathbbm{1}\mathbbm{1}^\top\right]( \tilde a_0 \bm z_0 - \tilde b_0 \bm v_0)\right\Vert^2  &= \left(\frac{\tilde a_0}{P} \mathbbm{1}^\top \bm z_0 - \frac{\tilde b_0}{P} \mathbbm{1}^\top \bm v_0\right)^2.
\end{align*}
Since $\bm z_0 \in \mathbb{R}^P$ is a standard Gaussian vector and $\tilde a_0 = O(1)$, we have $\tilde a_0 \mathbbm{1}^\top \bm z_0 / P \to 0$ as $P \to \infty$. Because $|v_0^\mu| \leq \varepsilon$ for all $\mu$, $|\mathbbm{1}^\top \bm v_0 / P| \leq \varepsilon$. Since $\tilde b_0 = O(1/\sqrt{P})$, we have $\tilde b_0 \mathbbm{1}^\top \bm v_0 / P \to 0$ also, so this whole term of $f_\tau(\bm z)$ can be eliminated for the purpose of computing the asymptotic capacity of our model. Similarly, for each $i \in \{1, 2, \dots, D\}$, 
\begin{align*}
     \frac{1}{P} \left\Vert \left[\frac{1}{P} \mathbbm{1}\mathbbm{1}^\top\right]( \bm z_i - \tilde b_1 \bm v_i)\right\Vert^2  &= \left(\frac{1}{P} \mathbbm{1}^\top \bm z_i - \frac{\tilde b_1}{P} \mathbbm{1}^\top \bm v_i\right)^2.
\end{align*}
Since $\bm z_i \in \mathbb{R}^P$ is a standard Gaussian vector, we have $\mathbbm{1}^\top \bm z_0 / P \to 0$. Since $|v_i^\mu| \leq \varepsilon$ for all $\mu$ and $\tilde b_1 = O(1/\sqrt{P})$, we also have $\tilde b_0 \mathbbm{1}^\top \bm v_i / P \to 0$, so we can eliminate this term from $f_\tau(\bm z)$ also. For the purpose of computing the asymptotic capacity of our model, we can therefore replace the field $f_\tau(\bm z)$ with a simplified field
$$\tilde f_\tau(\bm z) = \min_{\bm v \in \mathcal{B}(0)}\left[ \frac{1}{P}\left\Vert \Pi (a_0 \bm z_0 - b_0 \bm v_0)\right\Vert^2 + \sum_{i=1}^D  \frac{1}{P}\left\Vert \Pi (\bm z_i - b_1 \bm v_i)\right\Vert^2 \right],$$
just as in our analysis of correlated points. We will show that $\tilde f_\tau(\bm z)$ is asymptotically equivalent to the field $g_\tau(\bm z)$ 
$$g_\tau(\bm z) = \min_{\bm v \in \mathcal{B}(0)}\left[ \frac{1}{P}\left\Vert a_0 \bm z_0 - b_0 \bm v_0\right\Vert^2 + \sum_{i=1}^D  \frac{1}{P}\left\Vert \bm z_i - b_1 \bm v_i\right\Vert^2 \right]$$
that yields the capacity $\alpha_\mathrm{uncorr}^{-1}(\varepsilon, P) = \min_{\tau \geq 0} \mathbb{E}_{\bm z} g_\tau(\bm z)$
for the uncorrelated spheres model from the previous section with rescaled parameters
$$r_\mathrm{equiv} = r\sqrt{1-\psi}, \; \sigma_\mathrm{equiv} = \sigma\sqrt{1-\rho}, \textrm{ and } R_\mathrm{equiv} = R\sqrt{1-\gamma}.$$
The key will be to prove that, if $\bm v^*$ is the minimizer for the optimization defining $\tilde f_\tau(\bm z)$, then $\bar v_i^* = \mathbbm{1}^\top \bm v_i / P \to 0$ as $P \to \infty$ for all $i \in \{0, 1, \dots, D\}$.

Let us first show that $\bar v_0^* \to 0$ as $P \to \infty$: We know $\bm v_0^*$ is a minimizer for the optimization in the conditional field
$$\tilde f_\tau(\bm v_0 | \{\bm v_i^*\}_{i=1}^D) = \min_{(\bm v_0, \bm v_1^*, \dots, \bm v_D^*) \in \mathcal{B}(0)}\left[ \frac{1}{P}\left\Vert \Pi (a_0 \bm z_0 - b_0 \bm v_0)\right\Vert^2 \right] + \sum_{i=1}^D \frac{1}{P}\left\Vert \Pi (\bm z_i - b_1 \bm v_i^*)\right\Vert^2 .$$
Equivalently, $\bm v_0$ is a minimizer of $\|\Pi (a_0 \bm z_0 - b_0 \bm v_0)\|$ subject to
$$|v_0^\mu| \leq \varepsilon - \sqrt{\sum_{i=1}^D ((v^*_i)^\mu)^2} \leq \varepsilon.$$
\autoref{prop:vbar} then yields $\bar v_0^* \to 0$ as $P \to \infty$. To show that $\bar v_i^* \to 0$ with $i > 0$, we can similarly note that $\bm v_i^*$ is a minimizer of $\|\Pi(\bm z_i - b_1 \bm v_i)\|$ subject to
$$|v_i^\mu| \leq \sqrt{(\varepsilon - |(v_0^*)^\mu|)^2 - \sum_{j \neq 0, i} ((v^*_j)^\mu)^2} \leq \varepsilon.$$
Once again, this is in the form required by \autoref{prop:vbar}, which yields $\bar v_i^* \to 0$ as $P \to \infty$. Just as for the correlated points model, we can write
$$g_\tau(\bm z) =  \min_{\bm v \in \mathcal{B}(0)}\left[ \frac{1}{P}\left\Vert \Pi (a_0 \bm z_0 - b_0 \bm v_0)\right\Vert^2 + \sum_{i=1}^D  \frac{1}{P}\left\Vert \Pi (\bm z_i - b_1 \bm v_i)\right\Vert^2 + \left(\frac{\bm 1^\top}{P} (a_0 \bm z_0 - b_0 \bm v_0)\right)^2 + \sum_{i=1}^D \left(\frac{\bm 1^\top}{P}(\bm z_i - b_1 \bm v_i)\right)^2\right],$$
which asymptotically simplifies to
$$g_\tau(\bm z) =  \min_{\bm v \in \mathcal{B}(0)}\left[ \frac{1}{P}\left\Vert \Pi (a_0 \bm z_0 - b_0 \bm v_0)\right\Vert^2 + \sum_{i=1}^D  \frac{1}{P}\left\Vert \Pi (\bm z_i - b_1 \bm v_i)\right\Vert^2 + \left(b_0 \bar v_0\right)^2 + \sum_{i=1}^D \left(b_1 \bar v_i\right)^2\right].$$
Since minimizing the first two terms above automatically implies $\bar v_i \to 0$ for all $i$ as $P \to \infty$, we can drop the last two terms in this limit without affecting the value of $g_\tau(\bm z)$. Dropping the last two terms simply yields $\tilde f_\tau(\bm z)$, proving their asymptotic equivalence.

\subsection{A list of symmetries in the regression capacity for correlated spheres}
\begin{itemize}[itemsep=1mm, parsep=0pt]
    \item $(\varepsilon,\sigma,\tau)\leftrightarrow(a\varepsilon, a\sigma, \tau/a)$: Invariance to multiplying tolerance, label scale, and readout scale (equivalent to \emph{dividing} $\tau$) by some constant. In particular, setting $a = 1/\sigma$ reduces to the $\sigma = 1$ case. 
    This is exact even for finite $P$.
    \item $(r,R,\tau)\leftrightarrow(ar,aR,a\tau)$: Invariance to multiplying center norm and manifold radius and dividing readout scale (equivalent to \emph{multiplying} $\tau$) by some constant. In particular, setting $a = 1/r$ reduces to the $r = 1$ case. This is exact even for finite $P$.
    \item $(r,\psi)\leftrightarrow(r\sqrt{1-\psi},0)$: Adding uniform correlations among manifold centers is asymptotically equivalent to shrinking manifold center norms in the $P \to \infty$ limit.
    \item $(R,\gamma)\leftrightarrow(R\sqrt{1-\gamma},0)$: Adding uniform correlations among manifold axes is asymptotically equivalent to shrinking manifold radii in the $P \to \infty$ limit.
    \item $(\sigma, \rho)\leftrightarrow(\sigma\sqrt{1-\rho},0)$: Adding uniform correlations among labels is asymptotically equivalent to shrinking the label scale in the $P \to \infty$ limit
\end{itemize}

\section{Addition of an adaptive readout bias}\label{sec:SI bias}

When performing regression on real data, we typically want to allow a bias term $b$ in the linear readout so that $\hat y^\mu = w \cdot x^\mu + b$. For simplicity, we have so far ignored this. We will now sketch how to incorporate a bias into the definition \autoref{def:ddcap} and estimator \autoref{def:capformula} for instance-based capacity. We will study the bias-enabled regression capacity for our application the neural data (presented in \autoref{fig:hvm}), since this obviates the need for any data centering pre-processing step and simplifies the analysis.

Let us start by defining bias-enabled simulation capacity (analogous to \autoref{def:ddcap}):

\begin{definition}[Simulation capacity; bias-enabled]\label{def:ddcap-bias}
    Let $P$ manifolds $\{\mathcal{M}^\mu \subseteq \mathbb{R}^N\}_{\mu=1}^P$ and labels $\{y^\mu \in \mathbb{R}\}_{\mu=1}^P$ be given as above. For any tolerance $\varepsilon \geq 0$ and any matrix $M \in \mathbb{R}^{N_\mathrm{proj} \times N}$, we define the following concepts:
    \begin{itemize}
    \item (Admissible regression parameters) Define the set $\mathcal{A}_\mathrm{bias}(M)$ of admissible regression parameters (pairs of weights $
    \bm{w}$ and biases $b$) for the data transformed by $M$:
    $$\mathcal{A}_\mathrm{bias}(M) = \left\{(\bm{w}, b) \in \mathbb{R}^{N_\mathrm{proj}} \times \mathbb{R} : |\bm{w} \cdot (M\bm{x}^\mu) + b - y^\mu| \leq \varepsilon \textrm{ for all } \mu \in \{1, \dots, P\}, \bm{x}^\mu \in \mathcal{M}^\mu\right\}.$$
    \item (Regressible probability) Assume that the original data is regressible within the specified tolerance (i.e., that $\mathcal{A}_\mathrm{bias}(I_N) \neq \varnothing$) and that no constant readout is admissible (i.e., that $(0, b) \not\in \mathcal{A}_\mathrm{bias}(I_N)$ for all $b \in \mathbb{R}$). Then given any integer $N_\mathrm{proj}$ with $0 \leq N_\mathrm{proj} \leq N$, define
    $$p_\mathrm{bias}(N_\mathrm{proj}) = \mathbb{P}_{\Pi^{(N_\mathrm{proj})}}\left( \mathcal{A}_\mathrm{bias}(\Pi^{(N_\mathrm{proj})}) \neq \varnothing \right),$$
    where $\Pi^{(N_\mathrm{proj})} \in \mathbb{R}^{N_\mathrm{proj} \times N}$ is a uniformly random projection matrix.\footnote{By \enquote{uniformly random projection matrix}, we mean that $\Pi^{(N_\mathrm{proj})}$ is the first $N_\mathrm{proj}$ columns of an orthogonal matrix sampled from the normalized Haar measure on $O(N)$.}
    Note that $p_\mathrm{bias}(N_\mathrm{proj})$ is a non-decreasing function with $p_\mathrm{bias}(0) = 0$ and $p_\mathrm{bias}(N) = 1$. 
    \item (Critical dimension) Define the instance-based critical dimension of the dataset $\{\mathcal{M}^\mu\},\{y^\mu\}$ to be
    \[
    N^*_{\simcap,\mathrm{bias}}(\varepsilon) := \min_{p_\mathrm{bias}(N_\proj)\geq0.5}\{N_\proj\}
    \]
    \item (Simulation capacity) Finally, we define the instance-based regression capacity of the dataset $\{\mathcal{M}^\mu\},\{y^\mu\}$ to be 
    \[
    \alpha_{\mathrm{simulation},\mathrm{bias}}(\varepsilon) = \frac{P}{N_{\simcap,\mathrm{bias}}^*(\varepsilon)} \, .
    \]
    \end{itemize}
\end{definition}

This definition is conceptually no different than the bias-free case. We ask how many dimensions $N_\mathrm{proj}$ of our data we need to retain under a random projection so that an admissible linear readout \emph{with some adaptive bias that may differ for each random projection} exists with high probability. Like in our derivation of the bias-free instanced-based capacity estimator, we start by rewriting $p_\mathrm{bias}(N_\mathrm{proj})$ in terms of the probability that a random subspace $D \subseteq \mathbb{R}^N$ of dimension $N_\mathrm{proj}$ intersects some closed, convex cone (in particular, this is analogous to \autoref{lem:prob cone}):

Rather than working directly with $\mathcal{A}_\mathrm{bias}(M)$, the key here is to study the set of weights $\bm{w}$ for which some bias exists to make it into an admissible readout. In particular, we define $\tilde{\mathcal{A}}_\mathrm{bias}(M) = \{\bm{w} \in \mathbb{R}^{N_\mathrm{proj}} : \exists b \in \mathbb{R}, (\bm{w}, b) \in \mathcal{A}_\mathrm{bias}(M)\}$. Now $\mathcal{A}_\mathrm{bias}(M) \neq \varnothing$ is equivalent to $\tilde{\mathcal{A}}_\mathrm{bias}(M) \neq\varnothing$, so we can write
$$p_\mathrm{bias}(N_\mathrm{proj}) = \mathbb{P}_{\Pi^{(N_\mathrm{proj})}}\left(\tilde{\mathcal{A}}_\mathrm{bias}(\Pi^{(N_\mathrm{proj})}) \neq \varnothing\right).$$
In turn, $\tilde{\mathcal{A}}_\mathrm{bias}(\Pi^{(N_\mathrm{proj})}) \neq \varnothing$ is equivalent to $\operatorname{im}(\Pi^{(N_\mathrm{proj})}) \cap \tilde{\mathcal{A}}_\mathrm{bias} \neq \{0\}$, where like before we define $\tilde{\mathcal{A}}_\mathrm{bias} = \tilde{\mathcal{A}}_\mathrm{bias}(I_N)$. Since the subspace $\operatorname{im}(\Pi^{(N_\mathrm{proj})})$ is a uniformly random $N_\mathrm{proj}$-dimensional subspace of $\mathbb{R}^N$, we can replace it with the identically distributed subspace $QD$, where $Q$ is a uniformly random member of $O(N)$ and $D$ is some fixed $N_\mathrm{proj}$-dimensional subspace of $\mathbb{R}^N$. Finally, the condition $\tilde{\mathcal{A}}_\mathrm{bias} \cap QD \neq \{0\}$ is equivalent to $\tilde{\mathcal{A}}_\mathrm{bias}^+ \cap QD \neq \{0\}$, where $\tilde{\mathcal{A}}_\mathrm{bias}^+$ is the cone generated by $\tilde{\mathcal{A}}_\mathrm{bias}$. Thus,
$$p_\mathrm{bias}(N_\mathrm{proj}) = \mathbb{P}_Q\left(\tilde{\mathcal{A}}^+_\mathrm{bias} \cap QD \neq \{0\}\right),$$
and just like in the bias-free case, we have succeeded in rewriting $p_\mathrm{bias}$ as the probability of a closed, convex cone $\tilde{\mathcal{A}}_\mathrm{bias}^+$ having non-trivial intersection with a randomly-oriented subspace. We can therefore invoke the approximate conic kinematic formula \autoref{kinematic} to derive the following bias-enabled capacity estimator (analogous to \autoref{def:capformula}):

\begin{definition}[Capacity formula; bias-enabled]\label{def:capformula-bias}
    Let $P$ manifolds $\{\mathcal{M}^\mu \subseteq \mathbb{R}^N\}_{\mu=1}^P$ and labels $\{y^\mu \in \mathbb{R}\}_{\mu=1}^P$ be given as above. For any tolerance $\varepsilon \geq 0$, define an estimator $N^*_\mathrm{bias}(\varepsilon)$ for the critical dimension $N^*_\mathrm{bias}$ as
    \[
    N^*_\mathrm{bias}(\varepsilon) := \Exp_\bt\left[\min_{\bw^+\in\cA^+_\mathrm{bias}}\|\bt-\bw^+\|_2^2\right] \, ,
    \]
    where $\mathcal{A}_\mathrm{bias} = \mathcal{A}_\mathrm{bias}(I_N)$ is the set of regression weights for the original data \emph{such that some bias exists to make it into an admissible readout} and $S^+ := \{\lambda \mathbf{s} : \mathbf{s} \in S, \lambda \geq 0\}$ denotes the cone generated by a set $S \subseteq \mathbb{R}^N$. Define an estimator $\alpha_\mathrm{bias}(\varepsilon)$ for simulation capacity $\alpha_\simcap(\varepsilon)$ as
    \[
    \alpha_\mathrm{bias}(\varepsilon) := \frac{P}{N_\mathrm{bias}^*(\varepsilon)} \, .
    \]
\end{definition}

One can similarly incorporate an adaptive bias to our mean-field notion of regression capacity (originally defined in \autoref{sec:si mean-field}). For all of the synthetic data models we studied in the previous sections, one can show that the adaptive bias $b$ concentrates to a value of zero and therefore capacity is unaffected by the addition of a bias. This is essentially a consequence of the label distributions we studied being symmetric around 0.